\documentclass[twoside,12pt]{article}
\usepackage{amsmath,amssymb,amsthm,amsfonts}
\usepackage{paralist}
\usepackage{subfigure}
\usepackage[title]{appendix}

\usepackage{graphics} 
\usepackage{epsfig} 
\usepackage{graphicx}
\usepackage{caption}
\usepackage{epstopdf}
\usepackage[colorlinks,linkcolor=blue, citecolor=blue]{hyperref}
\usepackage{float}
\usepackage{pdfpages}
\usepackage{amsfonts}
\usepackage{amsmath,amsthm}
\usepackage{multirow}
\usepackage{txfonts}
\usepackage[numbers,sort&compress]{natbib}
\usepackage[font=scriptsize,labelfont=bf,width=0.95\textwidth]{caption}

\newtheorem{definition}{Definition}[section]
\newtheorem{theorem}{Theorem}[section]
\newtheorem{lemma}{Lemma}[section]
\newtheorem{corollary}{Corollary}[section]
\newtheorem{proposition}{Proposition}[section]
\newtheorem{remark}{Remark}[section]

\newcommand{\T}{{\rm Tr}}
\newcommand{\Rank}{{\rm rank~ }}

\newcommand{\R}{\mathbb R}

\newcommand{\bpp}{\begin{proposition}}
\newcommand{\epp}{\end{proposition}}

\newcommand{\bt}{\begin{theorem}}
\newcommand{\et}{\end{theorem}}
\newcommand{\bl}{\begin{lemma}}
\newcommand{\el}{\end{lemma}}
\newcommand{\bd}{\begin{definition}}
\newcommand{\ed}{\end{definition}}
\newcommand{\bc}{\begin{corollary}}
\newcommand{\ec}{\end{corollary}}
\newcommand{\bp}{\begin{proof}}
\newcommand{\ep}{\end{proof}}
\newcommand{\bx}{\begin{example}}
\newcommand{\ex}{\end{example}}
\newcommand{\bi}{\begin{exercise}}
\newcommand{\ei}{\end{exercise}}
\newcommand{\bo}{\begin{prop}}
\newcommand{\eo}{\end{prop}}
\newcommand{\br}{\begin{remark}}
\newcommand{\er}{\end{remark}}
\newcommand{\be}{\begin{equation}}
\newcommand{\ee}{\end{equation}}
\newcommand{\ba}{\begin{align}}
\newcommand{\ea}{\end{align}}
\newcommand{\bn}{\begin{enumerate}}
\newcommand{\en}{\end{enumerate}}
\newcommand{\bg}{\begin{align*}}
\newcommand{\bcs}{\begin{cases}}
\newcommand{\ecs}{\end{cases}}

\newcommand{\bean}{\begin{eqnarray*}}
\newcommand{\eean}{\end{eqnarray*}}

\allowdisplaybreaks

\numberwithin{equation}{section}

\begin{document}
\title{  Gagliardo-Nirenberg-Sobolev  inequalities and  ground states of Fermions in relativistic Hartree-Fock model}
\date{}
\author{
{ Yuan-da Wu\quad Xiaoyu Zeng
\quad Yimin Zhang}\thanks{
E-mail: wyuanda2021@126.com(Y. Wu); xyzeng@whut.edu.cn (X.Y. Zeng); zhangym802@126.com (Y.M. Zhang).
} \\
{\small\it  Center for Mathematical Sciences, Wuhan University of Technology,
 }\\
{\small\it  Wuhan 430070, P.R. China}\\\\
}

\maketitle



\vskip0.36in

\begin{abstract}

This paper presents a rigorous mathematical analysis of the relativistic Hartree-Fock model for finite Fermi systems. We first establish an optimal Gagliardo-Nirenberg-Sobolev (GNS) inequality with Hartree-type nonlinearities for orthonormal systems and characterize the qualitative properties of its optimizers. Furthermore, we derive a finite-rank Lieb-Thirring inequality involving convolution terms
and  show that it is  the duality of the  GNS-inequality-a result that, to our knowledge, has not previously appeared in the literature.  For the relativistic Hartree-Fock model,  we prove that ground states exist if and only if the coupling parameter $K<\mathcal{K}_\infty^{(N)}$, where $\mathcal{K}_\infty^{(N)}$  is the optimal constant in the GNS-inequality.  Finally, under suitable assumptions  on the external potentials, we calculate  the precisely asymptotic behavior of ground states as $K\nearrow\mathcal{K}_\infty^{(N)}$.

\end{abstract}
\noindent {\bf MSC}: 35J20; 35J60; 35Q55

\noindent {\bf Keywords}: Gagliardo-Nirenberg-Sobolev  inequality;  Lieb-Thirring inequality; finite Fermi system; ground state; asymptotic behavior


\vskip0.6in

\section{Introduction}
The relativistic Fermi systems subject to the gravitational interactions, such as neutron stars and white dwarfs, could be described by the following relativistic Hartree-Fock energy functional
\begin{equation}\label{HF0}
H_e=\sum_{i=1}^N\big[(p_i^2+m_i^2)^{\frac{1}{2}}-m_i\big]+\kappa\sum_{i>j}|x_i-x_j|^{-1},
\end{equation}
where the first term describes the relativistic kinetic energy and the last term describes gravitational ($\kappa<0$) or coulomb potential ($\kappa>0$).
In \cite{lt1}, Lieb and Thirring studied  the stability of Fermi systems under gravitational interactions.  Lieb and Yau in \cite{ly1} further considered the stability of (\ref{HF0}) with $m_i=m\not =0$ for all $i$ and $\kappa<0$. they revealed the relation between the critical constant $\kappa$ of stability  and  the Chandrasekhar limit. Later, in \cite{ly}, they extended this framework by introducing nuclear-electron correlation effects, thereby providing rigorous mathematical foundations for various stability criteria of ground state. Recently, Lenzmann and Lewin \cite{ll} studied the existence of minimizer in Hartree-Fock-Bogoliubov model, which provides a reliable description of unstable nuclei, introduces a pairing density matrix to describe the phenomenon of ``Cooper pairing''. Based on this work, Nguyen in \cite{n} analyzed the asymptotic behavior and showed that when $N$ large enough, up to scaling, the solution of the limiting mean field equation corresponds to the unique optimizer of the Hardy-Littlewood-Sobolev inequality. Very recently,  Chen, Guo, Nam and Ou Yang \cite{cgnoy} generalized the existence results of \cite{ll} to the critical mass case.  For more related studies, one can  refer to \cite{bfj, flss, l3, gl, cgwz,cg, cgl,ls} and references therein.

Under suitable simplifying assumptions, the Hartree-Fock functional can be  expressed as \cite{fgl, n}:
\begin{equation}\label{HF-1}
\mathcal{E}_K(\gamma):=\T\big((\sqrt{-\Delta+m^2}+V(x))\gamma\big)-K\int_{\mathbb{R}^3}(\rho_\gamma*|x|^{-1})\rho_\gamma dx,
\end{equation}
where $0\leq \gamma=\gamma^*$ is a compact operator in $L^2(\R^3;\mathbb{C})$ with $\Rank\gamma\leq N$ for some $N\in\mathbb{N}$  and  $\T(\gamma)<\infty$. $ \rho_\gamma\in L^1(\R^3)$ denotes the density associated with $\gamma$.
$V(x)\ge 0$ is a trapping potential and  the pseudo-differential operator $\sqrt{-\Delta+m^2}$ describes the kinetic energy of a fermion with mass $m>0$.   In general, the ground states of the relativistic Hartree-Fock energy functional can be obtained by solving the following constrained minimization problem:
 \begin{equation}\label{min-1}
\bar{E}_K(N)=\inf\Big\{\mathcal{E}_K(\gamma)\big|\ 0\le\gamma=\gamma^*\leq1,\  \T(\gamma)=N      \Big\}.
\end{equation}

By the spectral theory  a positive compact operator $\gamma$ in  $L^2(\mathbb{R}^3;\mathbb{C})$  can be diagonalized by
\begin{equation*}
\gamma=\sum_{i=1}^{\infty}n_i|u_i\rangle\langle u_i|,
\end{equation*}
where $n_i$ are non-negative constants and $\{u_i\}_{i=1}^\infty$ is the orthonormal basis of $L^2(\mathbb{R}^3;\mathbb{C})$.  As a consequence, 
\begin{equation*}
\text{Tr}(\sqrt{-\Delta+m^2}\gamma)=\sum_{i=1}^{\infty}n_i \|(-\Delta+m^2)^{\frac{1}{4}}u_i\|_2^2, \ \ \rho_\gamma=\sum_{i=1}^\infty n_i|u_i|^2 \text{ and }\T(\gamma)=\int_{\mathbb{R}^3}\rho_\gamma dx=\sum_{i=1}^{\infty}n_i.
\end{equation*}
We say that a compact self-adjoint operator $\gamma\geq0$ belongs to the $q$-trace class $\mathcal{S}^q$ ($1\le q\le \infty$), provided that 
\begin{equation*}
    \|\gamma\|_{\mathcal{S}^q}:=\begin{cases}
        (\sum_{i=1}^{\infty}n_i^q)^{\frac{1}{q}}<\infty ,\quad &\text{ if } 1\le q<\infty,\\
        \max_{i}n_i<\infty,\quad &\text{ if } q=\infty.
    \end{cases}
\end{equation*}
In particular, when $q=\infty$, we would denote $\|\cdot\|_{\mathcal{S}^\infty}$ as $\|\cdot\|$ for brevity in the whole paper.

To rigorously investigate minimization problems of the form \eqref{min-1} and related topics, it is essential to balance the kinetic energy  $\T(\gamma)$ and the interaction energy $\int_{\mathbb{R}^3}(\rho_\gamma*|x|^{-1})\rho_\gamma dx$.  For this purpose, we generalize the problem by   replacing the gravitational interaction $|x|^{-1}$ with a Riesz potential $|x|^{-\alpha}$, where $\alpha\in(0,2)$. We then study the  following  optimal Gagliardo-Nirenberg-Sobolev (GNS) inequality for  $\gamma\in\mathcal{S}^q$:
\begin{equation}\label{1.1}
\mathcal{K}_{\alpha,q}^{(N)}:=\inf_{
\begin{subarray}{c}
 \Rank\gamma\in [1,N],
\gamma\ge 0
\end{subarray}} \frac{\|\gamma\|_{\mathcal{S}^q}^{\frac{2-\alpha}{\alpha}
}\T (\sqrt{-\Delta}\gamma)}{\big(\int_{\mathbb{R}^3}\rho_\gamma(|x|^{-\alpha}*\rho_\gamma) dx\big)^{1/\alpha}}.
\end{equation}
We first remark that by applying the classical  Hardy-Littlewood-Sobolev inequality and the  Daubechies inequality \cite{liel}, which states that there exists $C>0$ such that
\begin{equation*}
\text{Tr}((-\Delta)^{\frac{1}{2}}\gamma) \ge C\int_{\mathbb{R}^3}\rho_\gamma^{\frac{4}{3}}dx,
\end{equation*}
one can deduce that $\mathcal{K}_{\alpha,q}^{(N)}>0$ is well-defined for each $N$. It is worth pointing out that some similar GNS-inequalities in a local setting have been established in two celebrated papers \cite{gln,fgl}.  Gontier,  Lewin and  Nazar in \cite{gln} studied the following  minimization problem in $L^2(\R^d)$  for dimensions $d\geq1$:
\begin{equation*}
J(N)=\inf\Big\{\T(-\Delta\gamma)-\frac{1}{p}\int_{\mathbb{R}^d}\rho_\gamma^p dx,\  0\leq\gamma=\gamma^*\leq1,\ \T(\gamma)=N \Big \},
\end{equation*}
where $1<p<\min\{2, 1+2/d\}$ is a mass-subcritical exponent. By applying a modified concentration compactness principle together with some refined estimates, they proved that there exists an  increasing sequence of integers
$N_1=1<N_2=2<N_3<\ldots<N_j<\ldots$,
for which the {\em binding inequality} (or strict subadditivity condition) holds for $J(N)$ with $N=N_j$. Consequently,  $J(N)$ admits a minimizer of the form $\gamma=\sum_{i=1}^{N}|u_i\rangle\langle u_i|$. In particular,  they proved that $\{u_i\}_{i=1}^N$ solves the following fermionic Nonlinear Schr\"odinger equations under {\em orthonormal conditions}
\begin{equation}\label{re1}
\Big[-\Delta- \Big(\sum_{n=1}^N |u_n|^2\Big)^{p-1}\Big]u_i=\mu_i u_i,\quad (u_i,u_j)_{L^2}=\delta_{ij},\quad  i,j=1,2,...,N,
\end{equation}
and $\{u_i\}_{i=1}^N$ is also an optimizer for the following GNS-inequality with orthonormal conditions:
\begin{equation*}
\mathcal{C}_{p,d}^{(N)}\Big(\int_{\mathbb{R}^d}\Big(\sum_{n=1}^N|u_n(x)|^2 \Big)^pdx \Big)^{\frac{2}{d(p-1)}}\le N^{\frac{2}{d(p-1)-1}} \sum_{n=1}^N\int_{\mathbb{R}^d}|\nabla u_n|^2dx, \ \forall u_i\in H^1(\mathbb{R}^d),\ (u_i,u_j)_{L^2}=\delta_{ij}.
\end{equation*}
Subsequently, Frank, Gontier and Lewin in \cite{fgl} further explored some similar problem, where the mass-critical exponent $p=1+ 2/d$ is particularly involved.  They  proved that  the best constant  $\mathcal{G}_{p,d}^{(N)}$  for the  following GNS-inequality 
\begin{equation}\label{GNS-Local}
\mathcal{G}_{p,d}^{(N)}\|\rho_\gamma\|_{L^p(\R^d)}^{\frac{2p}{d(p-1)}}\le \|\gamma\|_{\mathcal{S}^q}^\frac{p(2-d)+d}{d(p-1)}\T (-\Delta\gamma),\  \ \forall\ 0< \gamma=\gamma^* \text{ and  }\Rank \gamma\leq N
\end{equation}
is attained by some $\gamma=\sum_{i=1}^Rn_i|u_i\rangle\langle u_i|$ with $R\leq N$, where 
\begin{equation}\label{eq1.08}
\begin{split}
1\leq p\leq 1+\frac{2}{d}, \ \text{ and }\ q:=\begin{cases}
\frac{2p+d-pd}{2+d-pd}, & \text{ if }1\leq p< 1+\frac{2}{d},\\
+\infty, &\text{ if } p= 1+\frac{2}{d}.
\end{cases}
\end{split}
\end{equation}
Moreover, $\{u_i\}_{i=1}^N$ solves an orthonormal system similar to \eqref{re1}. Specially, it has been proved in \cite{fgl} that the inequality  \eqref{GNS-Local} is indeed  dual to the well-known finite rank
 {\em Lieb-Thirring  inequality} and  the quantitative estimates were obtained:
\begin{equation*}
\mathcal{G}_{p,d}^{(N)}(\mathcal{T}_{\kappa, d}^{(N)})^{\frac{2}{d}}=\Big(\frac{\kappa}{\kappa+\frac{d}{2}}\Big)^{\frac{2\kappa}{d}}\Big(\frac{d}{2\kappa+d}\Big) \text{ with } \kappa:=\frac{p}{p-1}-\frac{d}{2},
\end{equation*}
where $\mathcal{T}_{\kappa, d}^{(N)}>0$  is the best constant for the following finite rank Lieb-Thirring inequality
\begin{equation*}
\sum_{n=1}^N|\lambda_n(-\Delta +V)|^{\kappa}\le \mathcal{T}_{\kappa, d}^{(N)}\int_{\mathbb{R}^d}V(x)_-^{\kappa+\frac{d}{2}}dx, \ \text{for all }V\in L^{\kappa+\frac{d}{2}}(\mathbb{R}^d).
\end{equation*}
Here, $a_-=\max\{0,-a\}$ and $\lambda_n(-\Delta +V)< 0$ denotes the $n$-th negative eigenvalue of $-\Delta+V$ in $L^2(\mathbb{R}^d)$. 

The Lieb-Thirring inequality is a  famous inequality in literature of  mathematical physics, which has been extensively studied.  Lieb and Thirring \cite{lt0, lt} proved that if $\kappa>\frac12$ in $d=1$, or $\kappa>0$ in $d\geq2$, then there holds that  $$\mathcal{T}_{\kappa,d}:=\lim_{N\to\infty}\mathcal{T}_{\kappa,d}^{(N)}<\infty.$$
This estimate were further obtained in the critical cases for $\kappa=\frac12$ in $d=1$, and $\kappa=0$ in $d\geq3$  by \cite{c,lie0,r,we}, respectively. Moreover, to determine  the precise value of $\mathcal{T}_{\kappa,d}$ is a central issue in Density Functional Theory. One can refer to \cite{fgl,dll,lw}  for the recent progress on this aspect. 
 Very recently, Ilyin, Laptev and Zelik \cite{ilz} generalized this inequality to the bounded  domain and gave the sharp constant   when  the domain is sphere or torus.   

Motivated by the heuristic work of \cite{gln,fgl}, in this manuscript, we investigate the relativistic Hartree-type GNS-equality \eqref{1.1}. Specifically, we show that the $\mathcal{K}_{\alpha,q}^{(N)}$ can   be attained for all $q\in[1,+\infty]$ and $\alpha$ is even allowed to belong to  {\em mass-superciritical} regime.  More importantly, Through some variational arguments, we derive a  Lieb-Thirring type inequality involving convolution potentials, which has not been observed in the literature, to the best of our knowledge.  Our first main result addresses the achievement of  equality \eqref{1.1}  as follows:
\begin{theorem}\label{th1}
Let $1 \le N< \infty$, $0< \alpha< 2$ and $1\le q\le\frac{2-\alpha}{(1-\alpha)_+}$, where  $$\frac{2-\alpha}{(1-\alpha)_+}:=\begin{cases}
    \frac{2-\alpha}{1-\alpha},& \rm {for  }\ \ 0<\alpha<1, \\
    \infty, & \rm{for  }\ \ 1\le \alpha <2.
\end{cases}$$
Then, the best constant $\mathcal{K}_{\alpha,q}^{(N)}>0$  in   \eqref{1.1}
can be attained by some positive operator $\gamma$ with $1\leq {\rm rank~}\gamma\le N$.  Moreover, we have 
\begin{itemize}
    \item [\rm(i)] When $1\leq q<\frac{2-\alpha}{(1-\alpha)_+}$, then up to scaling,  every minimizing sequence for (\ref{1.1}) is  compact.  
    \item [\rm(ii)] Let $H_{\alpha,\gamma}:=\sqrt{-\Delta}-\frac{2}{\alpha}\rho_\gamma*|x|^{-\alpha}$, then $\gamma$ has the following explicit expression
\begin{equation*}
\gamma=\sum_{i=1}^Rk_i|u_{n_i} \rangle\langle u_{n_i}|,\ \ \text{for}\ k_i>0,\quad i=1,2,...,R\leq N,
\end{equation*}
where  $(\mu_{n_i},u_{n_i})$ with    $\mu_{n_i}\leq 0$ and $(u_{n_i},u_{n_j})_{L^2}=\delta_{ij}\ (i,j=1,\cdots R)$ are  eigen-pairs of $H_{\alpha,\gamma}$. Especially, if $R<N$, then $H_{\alpha,\gamma}$ has at most $R$ negative eigenvalues.
\item[ \rm(iii)] When $1< q<\frac{2-\alpha}{(1-\alpha)_+}$ or $\alpha=1$ and $q=\infty$, then $\mu_{n_i}<0$ for all $1\leq i\leq R$, and  $k_i$ can be expressed as 
\begin{equation*}
k_i= \begin{cases}
\frac{2-\alpha}{\alpha}(\sum_{k=1}^R|\mu_{n_k}|^{\frac{q}{q-1}})^{-1}|\mu_{n_i}|^{\frac{1}{q-1}}, &0<\alpha<2,\ 1<q<\frac{2-\alpha}{(1-\alpha)_+},\\
\frac{2-\alpha}{\alpha}(\sum_{k=1}^R|\mu_{n_k}|)^{-1}, &\alpha=1,\ q=\infty.
\end{cases}
\end{equation*}
In particular, $\mu_{n_i}(1\le i\leq R)$ are the first $R$  negative eigenvalues of $H_{\alpha,\gamma}$ provided that either $0<\alpha\le 1$, or $1<\alpha<2$ and  $\Rank \gamma<N$. 

\item[ \rm(iv)] 
If $0<\alpha\le1$, $q$ is close to $\frac{2-\alpha}{(1-\alpha)_+}$ enough, there exists a sequence $\{N_i\}_{i=1}^\infty$ with $\lim_{i\to\infty}N_i=\infty$, such that
\begin{equation*}
\mathcal{K}_{\alpha,q}^{(N_i)}<\mathcal{K}_{\alpha,q}^{(N_{i}-1)},\quad \forall\ i\ge 1.
\end{equation*}
Especially, there holds that 
\begin{equation}\label{eq1.8}
\mathcal{K}_{\alpha,q}^{(2N)}<\mathcal{K}_{\alpha,q}^{(N)},\  \ \forall \ N\in \mathbb{N}^+.
\end{equation}
\end{itemize}
\end{theorem}

The proof of the above theorem is primarily inspired by the arguments presented in \cite{fgl}. However, in our setting, the inclusion of the nonlocal operator $\sqrt{-\Delta}$ and Hartree nonlinearity introduces substantial difficulties in deriving the modified concentration-compactness principle. Another challenge in addressing inequality (\ref{1.1}) arises from the regularity theory for the fractional Laplacian. 
In general,  $(-\Delta)^s$ with $s=\frac{1}{2}$ is  called  ultra-relativistic Schr\"odinger operator, which  is a threshold from the aspect of regularity theory. That is, for a linear equation
\begin{equation*}
(-\Delta)^su+Vu=0, \text{ in }\mathbb{R}^d, \ d\geq1,
\end{equation*}
suppose  $u\in L^\infty(\mathbb{R}^d)$, if $s> \frac{1}{2}$ and $V(x)\in L^\infty(\mathbb{R}^d)$ one can show that $u\in C^{1,\beta}(\mathbb{R}^d)$ with some $\beta\in(0,1)$. However,   if $s\le \frac{1}{2}$, to obtain similar estimates one needs to further assume that  $V(x)\in C^{0,\gamma}(\mathbb{R}^d)$, for some $\gamma>1-2s$. We also emphasize that in $\mathbb{R}^3$, $\alpha=1$ ($\alpha=2$) corresponds to the mass-critical (Sobolev) exponent, therefore, our GNS-inequalities are  very general  and the parameters $q$ is  flexible. Specifically, our results cover the range $q\in [1,\frac{2-\alpha}{(1-\alpha)_+}]$ for all $0<\alpha<2$, whereas in \cite{fgl}, the paramter $p$ is assumed to at most  equal to the the mass critical exponent $ 1+\frac{2}{d}$, and $q$ takes the  specific values in  \eqref{eq1.08}.

\begin{remark}
\begin{itemize}
    \item [\rm (a)]Since   $\mathcal{K}_{\alpha,q}^{(N)}$ is non-increasing w.r.t. $N$, it then follows from the Hardy-Littlewood-Sobolev-inequality and Daubechies inequality that for the special case of  $q=1$, there holds
\begin{equation*}
\mathcal{K}_{\alpha,1}:=\lim_{N\to\infty}\mathcal{K}_{\alpha,1}^{(N)}=\inf_{N\ge 1} \mathcal{K}_{\alpha,1}^{(N)}>0.
\end{equation*}
Indeed, one can show that $\mathcal{K}_{\alpha,1}$ is actually  the best constant of the inequality
\begin{equation*}
\mathcal{K}_{\alpha,1}\Big(\int_{\mathbb{R}^3}({|x|^{-\alpha}}*\rho_\gamma)\rho_\gamma dx\Big)^{1/\alpha}\le \|\gamma\|_{\mathcal{S}^1}^{\frac{2-\alpha}{\alpha}}\rm Tr (\sqrt{-\Delta}\gamma)\  \text{ for all }0\leq \gamma \in \mathcal{S}^1.
\end{equation*}
\item[\rm(b)]When $N=1$, (\ref{1.1}) degenerates to the following classical Hartree-type Gagliardo-Nirenberg inequality in $H^\frac{1}{2}(\mathbb{R}^3)$
\begin{equation*}
\mathcal{K}_{\alpha,q}^{(1)}\Big(\int_{\mathbb{R}^3}({|x|^{-\alpha}}*|u|^2)|u|^2 dx\Big)^{1/\alpha}\le \|u\|_2^{\frac{4-2\alpha}{\alpha}}\|(-\Delta)^{\frac{1}{4}}u\|_2^2, \ \forall\ u\in H^\frac{1}{2}(\mathbb{R}^3).
\end{equation*}
\end{itemize}

\end{remark}

We next establish a type of finite rank Lieb-Thirring inequality containing convolutions, which is dual to the GNS-inequality \eqref{1.1}. Our result can be stated as follows.
\begin{theorem}\label{dual} \rm{(Duality).}
        Let $1\le N<\infty$, $0<\alpha<2$ and $1<q\le \infty$,  then
we have  the following optimal  Hartree type Lieb-Thirring inequality:
\begin{equation}\label{re3}
   \sum_{n=1}^N|\lambda_n(\sqrt{-\Delta}+V(x)*|x|^{-\alpha})|^{q'}\le \mathcal{L}_{\alpha,q'}^{(N)}\Big(\int_{\mathbb{R}^3}(V_-(x)*|x|^{-\alpha})V_-(x)dx\Big)^\frac{q'}{2-\alpha}
\end{equation}
holds for all $V(x)$ satisfying $\int_{\mathbb{R}^3}(|V(x)|*|x|^{-\alpha})|V(x)|dx< \infty$, where  $q'=\frac{q}{q-1}$, $\mathcal{L}_{\alpha,q'}^{(N)}<\infty$  is the best constant, and $\lambda_n(\sqrt{-\Delta}+V(x)*|x|^{-\alpha})$ is the $n$-th negative eigenvalue of $\sqrt{-\Delta}+V(x)*|x|^{-\alpha}$ when it exists and $0$ otherwise.
Moreover,  the following identity holds
\begin{equation*}
     \mathcal{K}_{\alpha, q}^{(N)}(\mathcal{L}_{\alpha,q'}^{(N)})^{\frac{(2-\alpha)(q-1)}{\alpha q}}=\frac{\alpha}{2}\Big(\frac{2-\alpha}{2}\Big)^{\frac{2-\alpha}{\alpha}},
\end{equation*}
where $\mathcal{K}_{\alpha, q}^{(N)}$ is the best constant in (\ref{1.1}). 
\end{theorem}

\begin{remark}
The classical Lieb-Thirring inequality for $\sqrt{-\Delta}$ states that for $V(x)\in L^{q'+3}(\mathbb{R}^3)$, then
\begin{equation}\label{pr}
\sum_{n=1}^N|\lambda_n(\sqrt{-\Delta} +V)|^{q'}\le \tilde {\mathcal{L}}_{\alpha,q'}^{(N)}\int_{\mathbb{R}^3}V(x)_-^{q'+3}dx.
\end{equation}
If substituting $V(x)*|x|^{-\alpha}$ into \eqref{pr}, we arrive at
\begin{equation}\label{re4}
\sum_{n=1}^N\Big|\lambda_n(\sqrt{-\Delta}+V(x)*|x|^{-\alpha})\Big|^{q'}\le \tilde {\mathcal{L}}_{\alpha,q'}^{(N)}\int_{\mathbb{R}^3}\Big(V(x)_-*|x|^{-\alpha}\Big)^{q'+3}dx.
\end{equation}
Unfortunately, we cannot finger out whether the inequality (\ref{re3}) or (\ref{re4}) is superior. 
\end{remark}

In what  follows, we are  concerned with the existence  and quantitative properties of minimizers for the relativistic Hartree-Fock energy functional.
 Instead of studying (\ref{min-1}) directly, we restrict  ourself to the case of  finite particles for simplicity. We consider the following constrained  minimizing problem  \begin{equation}\label{mini-N}
E_K(N)=\inf\Big\{\mathcal{E}_K(\gamma)\big| \gamma\in \Gamma      \Big\},
\end{equation}
where the constrained manifold $\Gamma$ is defined as
\begin{equation*}
\Gamma:=\Big \{\gamma=\sum_{i=1}^r|u_i\rangle \langle u_i|,\  \forall u_i\in \mathcal{H},\ (u_i,u_j)_{L^2}=\delta_{ij},\ \forall 1\le i,j\le r, 0<{\rm rank~}\gamma=r\le N  \Big \}.
\end{equation*}
Here  to  ensure all terms in \eqref{HF-1} make sense for $V(x)\ge 0$, we introduce  the inner  space  
\begin{equation*}
\mathcal{H}:=\Big \{u\in H^{\frac{1}{2}}(\mathbb{R}^3), \int_{\mathbb{R}^3}V(x)|u|^2dx< \infty \Big \}.
\end{equation*}

 Before stating our main results, we first define ground states of a Fermi system with potentials according to the Aufbau principle.
\begin{definition} (Ground state) In a $L$-Fermi system,  $(u_1,u_2,...,u_L)\in L^2(\mathbb{R}^3;\mathbb{C}^L)$ is called a ground state of the the following system, if  $(u_i,u_j)_{L^2}=\delta_{ij}$ for all $i,j=1,2,\cdots,L$ and  
\begin{equation}\label{1.3}
H_V u_i:=\Big(\sqrt{-\Delta+m^2}+V(x)-2K\big(\sum_{k=1}^L u_k^2\big)*|x|^{-1}\Big)u_i=\mu_i u_i, \ i=1,2,\cdots, L,
\end{equation}
 where $\mu_1<\mu_2\le...\le\mu_L$ are the first-$L$ eigenvalues of the operator $H_V$.
\end{definition}

We now state our results on the minimization problem \eqref{mini-N},  and  show that there exists a threshold for the existence of minimizers, which is indeed  the ground states of the system \eqref{1.3}.

\begin{theorem}\label{th2}
Let $N\in\mathbb{N}^+$ be fixed and $\mathcal{K}_\infty^{(N)}:=\mathcal{K}_{1,\infty}^{(N)}$ be the best constant given by (\ref{1.1}). Assume that $V(x)$ is a trapping potential satisfying
\begin{equation}\label{1.4}
0\le V(x)\in C^{1}(\mathbb{R}^3),\quad \lim_{|x|\to\infty}V(x)=\infty \ \text{ and }\ \inf_{x\in\mathbb{R}^3}V(x)=0.
\end{equation}
Then
\begin{itemize}
    \item [\rm(i)]For $K\in \big(0,\mathcal{K}_\infty^{(N)}\big)$, problem (\ref{mini-N}) has at least one minimizer $\gamma_K=\sum_{i=1}^r|u_i\rangle\langle u_i|$,  where  $(u_1,u_2,...,u_r)$ is a ground state of $r$-Fermi system (\ref{1.3}).

    \item[\rm(ii)] When $K$ is close to $ 0^+$, then (\ref{mini-N}) has a unique minimizer $\gamma_K$ and $\Rank \gamma_K=1$; when $K$ is close to $(\mathcal{K}_\infty^{(N)})^-$, then  any minimizer $\gamma_K$ for (\ref{mini-N}) satisfies $\Rank \gamma_K\ge [\frac{N}{2}]+1$.  

    \item[\rm(iii)]For $K\in \big[\mathcal{K}_\infty^{(N)},\infty\big)$, there is no minimizer for the problem (\ref{mini-N}).
\end{itemize}

\end{theorem}

The primary challenge in establishing the nonexistence result stems from the technical selection of test functions, particularly due to the identical scaling rates between the Hartree term and the fractional operator term  as well as maintaining the orthonormality. For brevity, we only  focus  on the case of $q=\infty$. Nevertheless, our arguments can be  extend to finite but sufficiently large $q$, which, however, requires much more careful calculations. We also note that, in the context of bosonic systems, numerous studies have addressed the existence and asymptotic behavior of ground states. A comprehensive review of these works falls outside the scope of this paper. Interested readers may refer to \cite{bc, ly1,ly, gs, es, fjl, gz,gz1,yy} and references therein for further details. 






Based on the existence and non-existence results in Theorem \ref{th2}, we next study  the asymptotic behavior of minimizers for \eqref{mini-N} as $K\nearrow \mathcal{K}_\infty^{(N)}$.
\begin{theorem}\label{th3}
Assume that $V(x)$ satisfies (\ref{1.4}), and  denote 
\begin{equation}\label{eq-1.15}
    \Lambda:=\big\{x\in\mathbb{R}^3: V(x)=0\big\}.
\end{equation}
Let 
$\gamma_k:=\gamma_{K_k}=\sum_{i=1}^{r} |u_i^k\rangle \langle u_i^k|$  be a minimizer of (\ref{mini-N}) for each $K_k\nearrow \mathcal{K}_\infty^{(N)}$, as $k\to \infty$, where   $[\frac{N}{2}]+1\le r\le N$ is an integer. Then, 
\begin{equation}\label{eq-eps}
\varepsilon_k:=[{\rm Tr} (\sqrt{-\Delta}\gamma_{k})]^{-1} \to 0^+\ \text{ as } K_k\nearrow \mathcal{K}_\infty^{(N)}
\end{equation}
and up to subsequence, there hold
\begin{equation}\label{eq-gamma0}
\tilde \gamma_k=\sum_{i=1}^r|\varepsilon_k^{\frac{3}{2}} u_i^{k}(\varepsilon_k x+z_k)\rangle\langle \varepsilon_k^{\frac{3}{2}} u_i^{k}(\varepsilon_k x+z_k) |\overset{\star}\rightharpoonup \gamma  \quad \text{in}\ \mathcal{S}^{1}, \ \text{ as } K_k\nearrow \mathcal{K}_\infty^{(N)}
\end{equation}
and 
\begin{equation}\label{eq-zero}
\lim_{k\to\infty}\T(\sqrt{-\Delta}\tilde \gamma_k)=\T(\sqrt{-\Delta}\gamma),\quad \lim_{k\to\infty}\int_{\mathbb{R}^3}(\rho_{\tilde \gamma_k}*|x|^{-1})\rho_{\tilde \gamma_k}dx=\int_{\mathbb{R}^3}(\rho_{ \gamma}*|x|^{-1})\rho_{ \gamma}dx,
\end{equation}
where $\lim_{k\to\infty}z_k=z_0\in \Lambda$ and  $\gamma=\sum_{i=1}^R |Q_i\rangle \langle Q_i|$  ($R\le r$) is an optimizer of (\ref{1.1}). 
Moreover, if $r=R$, then $\tilde \gamma_k\overset{k}\to \gamma$ in $\mathcal{S}^1\cap\mathcal{S}^\infty$ and 
\begin{equation}\label{eq-zero2}
\varepsilon_k^\frac{3}{2} u_i^{k}(\varepsilon_k x+z_k)\overset{k}\to Q_i(x)\ \text{ strongly in }\ H^{\frac{1}{2}}(\mathbb{R}^3),  \   \text{ for all } \ i=1,2,\cdots,R.
\end{equation}

\end{theorem}

\begin{remark}
From (iv) in Theorem \ref{th1},  we can deduce that there exits   an increasing sequence $\{ N_i\}_{i=1}^\infty$ with $\lim_{i\to\infty}{ N_i}=\infty$ such that $\mathcal{K}_\infty^{( N_{i})}<\mathcal{K}_\infty^{( N_{i}-1)}$ for all $i=1,2,\cdots$. If we restrict $N= N_i$ in Theorem \ref{th3}, then one can see that  $r=R= N_i$. Actually, in our Theorem \ref{th4} below, we shall  show that   $r=R$  always holds for any $N\in\mathbb{N}^+$ provided  more information on $V(x)$ is given. 
\end{remark}

In what follows, we consider special potentials whose local expansions around their minimal points are known. For such potentials, we explicitly compute the energy $E_{K}(N)$ and the  blow-up rate of minimizers for \eqref{mini-N}  as $K\nearrow  \mathcal{K}_\infty^{(N)}$. More importantly,  we prove that  $r=R$ always holds for any $N\in\mathbb{N}^+$. Some of our arguments are inspired by the ideas in \cite{gzz,gwzz}.


\begin{theorem}\label{th4}
Let $\varepsilon_k$,  $z_0$  and $z_k$ be given by Theorem \ref{th3}. Assume that $V(x)$ satisfies (\ref{1.4}) and
\begin{equation}\label{1.5}
V(x)=h(x)\prod_{j=1}^l|x-x_j|^{p_j}\ \text{ with $0<h(x)\in C^1(\mathbb{R}^3)$  and $0<p_j<1$ for all $1\le j\le l$. }
\end{equation}
Denote $p=\max_{1\leq j\leq l}p_j$ and 
\begin{equation}\label{eq-1.17}
\mathcal{Z}=\{x_j| \iota_j=\iota\} \text{ where }\iota=\min_{1\le j\le l}\iota_j\ \text{ with }\ \iota_j=\lim_{x\to x_j}\frac{V(x)}{|x-x_j|^p}\in(0,\infty].
\end{equation}
For $\gamma\in\mathcal{S}^1$ given in Theorem \ref{th3},  set
\begin{equation}\label{eq-1.18}
    \bar\Gamma:=\{y\in\mathbb{R}^3: \int_{\mathbb{R}^3}|x+y|^p\rho_\gamma(x)dx=\bar\kappa\} \text{ with }\bar \kappa:=\inf_{y\in\mathbb{R}^3}\int_{\mathbb{R}^3}|x+y|^p\rho_\gamma(x)dx .
\end{equation}
Then, we have $R=r$ and thus \eqref{eq-zero2} holds. Specifically, there also hold
 \begin{equation}\label{eq-energy}
\begin{aligned}
E_{K_k}(N)=\big(1+o_k(1)\big)\frac{p+1}{p} (p\iota\bar\kappa)^{\frac{1}{p+1}}\Big(\int_{\mathbb{R}^3}(\rho_{\gamma}*|x|^{-1})\rho_{ \gamma}dx\Big)^{\frac{p}{p+1}}(\mathcal{K}_\infty^{(N)}-K_k)^{\frac{p}{p+1}} \ \text{as}\  K_k\nearrow \mathcal{K}_\infty^{(N)},
\end{aligned}
\end{equation}
\begin{equation}\label{eq-vas}
\varepsilon_{k}=\big(1+o_k(1)\big)\Big[(p\iota\bar \kappa)^{-1}\int_{\mathbb{R}^3}(\rho_{\gamma}*|x|^{-1})\rho_{ \gamma}dx\big(\mathcal{K}_\infty^{(N)}-K_k\big)\Big]^{\frac{1}{p+1}}\to0^+ \ \text{as}\  K_k\nearrow \mathcal{K}_\infty^{(N)}
\end{equation}
 and 
\begin{equation}\label{eq-points}
   z_0\in\mathcal{Z}\subset\Lambda \ \text{ and }\ \lim_{k\to\infty}\frac{z_k-z_0}{\varepsilon_k}=y\in \bar\Gamma.
\end{equation}
\end{theorem}

\textbf{Structure of the paper.} 
In Section \ref{Sec-2}, we prove Theorem \ref{th1} by establishing the optimal Hartree-type GNS-inequality  \eqref{1.1}  and analyzing the qualitative properties of its optimizers. Then, we  derive  its duality  in Theorem \ref{dual}.
In Section \ref{Sec-3}, based on this inequality, we show that  there exists  a  threshold to  distinguish the existence and nonexistence of ground states for problem \eqref{mini-N}, which finishes the proof of  Theorem \ref{th2}.
Section \ref{Sec-4}  is devoted to studying   the asymptotic behavior of ground states as the parameter $K$ approaches the  threshold $\mathcal{K}_\infty^{(N)}$, and we complete the proofs of Theorems \ref{th3} and  \ref{th4}.
In the Appendix, we prove that  optimizers of the GNS-inequality exhibit polynomial decay at infinity, and  some Pohozaev type identities for these optimizers are also  derived, which, to the best of our knowledge, has not been previously captured  in the literature.


\section{Existence of optimizers and dual version for Hartree type inequality} \label{Sec-2}
This section is devoted to the proof of Theorem \ref{th1}. Firstly, we show the existence of optimizer for $\mathcal{K}_{\alpha,q}^{(N)}$ by employing some variational techniques, such as the concentration compactness principle. Then, we establish some  analytic properties of   optimizers  by applying some regularity  theories.

\noindent\textit{ Proof of Theorem \ref{th1}.} \textbf{I. Existence of optimizers.} Assume that there exists a minimizing sequence $\{\gamma_n\}$ satisfying 
\begin{equation}\label{eq2.000}
    \gamma_n:=\sum_{i=1}^N a_{in}|u_{in}\rangle \langle u_{in}|,\ a_{in}\ge 0, \ {\rm rank~}\gamma_n\le N \ \text{ for all $n$},
\end{equation}
where we take $a_{in}=0$ alternatively for $i>{\rm rank~}\gamma_n$ provided  ${\rm rank~}\gamma_n<N$.  Normalizing the sequence such that 
\begin{equation}\label{eq2.0}
\text{Tr}(\sqrt{-\Delta}\gamma_n)=1,\quad \|\gamma_n\|_{\mathcal{S}^q}=1.
\end{equation}
Denote the density function $\rho_n:=\rho_{\gamma_n}$, then,
\begin{equation}\label{2.00}
\lim_{n\to\infty}\int_{\mathbb{R}^3}\big(|x|^{-\alpha}*\rho_n\big)\rho_n dx=\big(\mathcal{K}_{\alpha,q}^{(N)}\big)^{-\alpha}.
\end{equation}
From the above normalization we obtain
\begin{equation*}
\|\gamma_n\|\le\|\gamma_n\|_{\mathcal{S}^q}=1
\end{equation*}
and
\begin{equation*}
\int_{\mathbb{R}^3}\rho_n(x)dx=\text{Tr}(\gamma_n)=\sum_{i=1}^N a_{in}\le \sum_{i=1}^N\|\gamma_n\|\le N.
\end{equation*}
According to Hoffmann-Ostenhof type inequality (see e.g., Lemma 2.1 in \cite{ll}), we obtain
\begin{equation*}
1=\text{Tr}(\sqrt{-\Delta}\gamma_n)\ge \int_{\mathbb{R}^3}|(-\Delta)^{\frac{1}{4}}\sqrt{\rho_n}|^2dx.
\end{equation*}
Therefore, $\{\sqrt{\rho_n}\}_n$ is uniformly bounded in $H^{\frac{1}{2}}(\mathbb{R}^3)$. We can extract a  subsequence, still denoted by $\{\sqrt{\rho_n}\}_n$, such that
\begin{equation*}
\begin{aligned}
&\sqrt{\rho_n}\overset{n}\rightharpoonup \sqrt{\rho}\ \text{ weakly  in }\ H^\frac{1}{2}(\mathbb{R}^3),
\text{ and } \sqrt{\rho_n}\to \sqrt{\rho}\ \text{strongly in}\ L_{{\rm loc}}^s(\mathbb{R}^3)\ \text{ for }\  2\le s<3.
\end{aligned}
\end{equation*}
Next, we intend to verify that there exist $\bar R>0$, $C_{\bar R}>0$ and $\{y_n\}_{n=1}^\infty \subset \mathbb{R}^3$ such that 
\begin{equation}\label{eq2.05}
    \liminf_{n\to\infty}\int_{B_{\bar R}(y_n)}\rho_n(x) dx\geq C_{\bar R}>0.
\end{equation}
If it fails,  by  the vanishing lemma in  \cite{l}, we have
$\rho_n \overset{n}\to 0  \text{ in } L^r(\mathbb{R}^3)$, for all $\ r\in(1,3/2)$.
This indicates that 
\begin{equation*}
\int_{\mathbb{R}^3}\Big(\rho_n*{|x|^{-\alpha}}\Big)\rho_n dx\to 0\ \text{as} \ n\to \infty,
\end{equation*}
which contradicts  (\ref{2.00}).

From \eqref{eq2.05} we see that  there exists a $0\not=\sqrt\rho\in H^\frac12(\mathbb{R}^3)$ such that, up to a subsequence, 
\begin{equation}\label{eq2.06}
  \sqrt{\rho_n(\cdot-y_n)}\overset n\rightharpoonup \sqrt{\rho}\not =0 \text{ weakly in $H^{\frac{1}{2}}(\mathbb{R}^3)$}.  
\end{equation}
Moreover, we can apply the Banach-Alaoglu Theorem to deduce that $\{\gamma_n\}$ has  a weak-$\ast$ limit   in the trace class topology, i.e. $\gamma_n\overset{\star}\rightharpoonup \gamma\not =0$, and density function of $\gamma$ satisfies $\rho_\gamma=\rho$.

Since the minimization problem \eqref{1.1} is invariant up to translations, we may assume that $y_n\equiv0$ in \eqref{eq2.05}, and then deduce from \eqref{eq2.06} that  there exists a sequence $\{R_n\}_{n=1}^\infty$ with $R_n\overset{n}{\to}\infty$, such that
\begin{equation*}
\lim_{n\to\infty}\int_{|x|\le R_n}\rho_n(x)dx=\int_{\mathbb{R}^3}\rho(x)dx \ \text{ and  }\ \lim_{n\to\infty}\int_{R_n\le |x|\le 6R_n}\rho_n(x)dx=0.
\end{equation*}

Let $\chi\in C_c^\infty\big(\mathbb{R}^3,[0,1]\big)$ satisfy
\begin{equation*}
\chi(x)\equiv 1, \ \text{for}\ |x|<1;\quad \chi(x)\equiv 0,\ \text{for}\ |x|\ge 2.
\end{equation*}
Define $\chi_n(x):=\chi(\frac{x}{R_n})$ and $\eta_n(x):=\sqrt{1-\chi_n^2}$. Then,  $$\chi_n^2\rho_n\overset{n}\to \rho \text{ in $L^1(\mathbb{R}^3)\cap L^r(\mathbb{R}^3)$ for $r\in(1,3/2)$\ \ and\ \ $(|\nabla \chi_n|^2+|\nabla \eta_n|^2)\rho_n\overset{n}\to 0$ in $L^1(\mathbb{R}^3)$.} $$

Using the IMS type formula in \cite{ly, n} and Fatou's lemma for operators \cite{s}, we have
\begin{equation}\label{eq-trace1}
\begin{aligned}
\text{Tr}(\sqrt{-\Delta}\gamma_n)&\ge \text{Tr}(\sqrt{-\Delta}\chi_n\gamma_n\chi_n)+\text{Tr}(\sqrt{-\Delta}\eta_n\gamma_n\eta_n)+O(R_n^{-1})\\
&\ge \text{Tr}(\sqrt{-\Delta}\gamma)+\text{Tr}(\sqrt{-\Delta}\eta_n\gamma_n\eta_n)+o_n(1).
\end{aligned}
\end{equation}
Moreover,
\begin{equation}\label{2.1}
\begin{aligned}
&\int_{\mathbb{R}^3}\big(\rho_n*|x|^{-\alpha}\big)\rho_ndx=\int_{\mathbb{R}^3}\int_{\mathbb{R}^3}\frac{(\chi_n^2\rho_n(x)+\eta_n^2\rho_n(x))(\chi_n^2\rho_n(y)+\eta_n^2\rho_n(y))}{|x-y|^{\alpha}}dxdy\\
&=\int_{\mathbb{R}^3}\int_{\mathbb{R}^3}\frac{\chi_n^2\rho_n(x)\chi_n^2\rho_n(y)}{|x-y|^{\alpha}}dxdy+\int_{\mathbb{R}^3}\int_{\mathbb{R}^3}\frac{\eta_n^2\rho_n(x)\eta_n^2\rho_n(y)}{|x-y|^{\alpha}}dxdy+2\int_{\mathbb{R}^3}\int_{\mathbb{R}^3}\frac{\chi_n^2\rho_n(x)\eta_n^2\rho_n(y)}{|x-y|^{\alpha}}dxdy.
\end{aligned}
\end{equation}
To estimate the interaction term, we define $\tilde \chi_{n}(x):=\chi(\frac{x}{3R_n})$, $\tilde \eta_{n}(x)=\sqrt{1-\chi^2(\frac{x}{3R_n})}$,  and  divide $\eta_n(y)$ into two terms
\begin{equation*}
\eta_n^2(y)=[\eta_n^2(y)-\tilde \eta_{n}^2(y)]+\tilde \eta_n^2(y).
\end{equation*}
Inserting it into the interaction term, we get
\begin{equation}\label{2.2}
\chi_n^2(x)|x-y|^{-\alpha}\tilde \eta_{n}^2(y)\le \frac{\chi_n^2(x)\tilde \eta_{n}^2(y)}{\big(|y|-|x|\big)^{\alpha}}\le \frac{1}{R_n^{\alpha}}
\end{equation}
and
\begin{equation}\label{2.3}
\chi_n^2(x)|x-y|^{-\alpha}(\eta_n^2(y)-\tilde \eta_{n}^2(y))\le {\mathbb{I}}_{\{R_n\le|y|\le 6R_n\}}|x-y|^{-\alpha}.
\end{equation}
Taking (\ref{2.2}) and (\ref{2.3}) into consideration, we infer from Hardy-Littlewood-Sobolev inequality that
\begin{equation*}
\begin{aligned}
\int_{\mathbb{R}^3}\int_{\mathbb{R}^3}\frac{\chi_n^2\rho_n(x)\eta_n^2\rho_n(y)}{|x-y|^{\alpha}}dxdy\le& \frac{1}{R_n^{\alpha}}\int_{\mathbb{R}^3}\int_{\mathbb{R}^3}\rho_n(x)\rho_n(y)dxdy
+\int_{\mathbb{R}^3}\int_{R_n\le|y|\le 6R_n}\frac{\rho_n(x)\rho_n(y)}{|x-y|^{\alpha}}dxdy\\
\le& \frac{\|\rho_n\|_1^2}{R_n^{\alpha}}+C\|\rho_n\|_{\frac{6}{6-\alpha}}\Big(\int_{R_n\le|y|\le 6R_n}\rho_n^{\frac{6}{6-\alpha}}(y)dy\Big)^{\frac{6-\alpha}{6}}=o_n(1).
\end{aligned}
\end{equation*}
This together with \eqref{2.1} indicates that 
\begin{equation}\label{eq2.8}
\int_{\mathbb{R}^3}\big(\rho_n*|x|^{-\alpha}\big)\rho_ndx
=\int_{\mathbb{R}^3}\int_{\mathbb{R}^3}\frac{\chi_n^2\rho_n(x)\chi_n^2\rho_n(y)}{|x-y|^{\alpha}}dxdy+\int_{\mathbb{R}^3}\int_{\mathbb{R}^3}\frac{\eta_n^2\rho_n(x)\eta_n^2\rho_n(y)}{|x-y|^{\alpha}}dxdy+o_n(1).
\end{equation}
The proof for the existence of optimizers would be  divided  into two cases: Case (a)  $0<\alpha<1$ and  $1\le q\le \frac{2-\alpha}{1-\alpha}$ and Case (b) $1\le\alpha<2$ and $q\in[1,\infty]$. 

Case (a) $0<\alpha<1$ and  $1\le q\le \frac{2-\alpha}{1-\alpha}$. From the renormalization \eqref{eq2.0} 
\begin{equation}\label{eq2.4}
1=\|\gamma_n\|_{\mathcal{S}^q}^{2-\alpha}(\T(\sqrt{-\Delta}\gamma_n))^{\alpha}=\big(\T(\gamma_n^q)\big)^{\frac{2-\alpha}{q}}\big(\T(\sqrt{-\Delta}\gamma_n)\big)^{\alpha}.
\end{equation}
Recall that if $1<q<\infty$, then,
\begin{equation}\label{eq2.5}
\begin{aligned}
\T(\gamma_n^q)&=\T\Big((\chi_n^2+\eta_n^2)\gamma_n^q\Big)\ge \T(\chi_n^{2q}\gamma_n^q) +\T(\eta_n^{2q}\gamma_n^q) = \T(\chi_n^q\gamma_n^q\chi_n^q)+\T(\eta_n^q\gamma_n^q\eta_n^q)\\
&\ge \T\big((\chi_n\gamma_n\chi_n)^q\big)+\T\big((\eta_n\gamma_n\eta_n)^q\big)\ge \T(\gamma^q)+\T\big((\eta_n\gamma_n\eta_n)^q\big)+o_n(1).
\end{aligned}
\end{equation}
Moreover, since  for $\alpha+\theta\ge 1$,  there holds  that
\begin{equation*}
a^\alpha c^\theta+b^\alpha d^\theta\le (a+b)^\alpha(c+d)^\theta \text{ for  any }a,b,c,d\geq0.
\end{equation*}
Since $0<\alpha<1$ and  $1\le q\le \frac{2-\alpha}{1-\alpha}$,  we deuce from \eqref{eq-trace1} and  \eqref{eq2.5} that 
\begin{equation*}
\begin{aligned}
1&\ge \big[\T(\gamma^q)\big]^{\frac{2-\alpha}{q}}\big[\T(\sqrt{-\Delta}\gamma\big]^{\alpha}+\big[\T\big((\eta_n\gamma_n\eta_n)^q\big)\big]^{\frac{2-\alpha}{q}}\big[\T(\sqrt{-\Delta}\eta_n\gamma_n\eta_n)\big]^\alpha+o_n(1)\\
&\ge [\T(\gamma^q)]^{\frac{2-\alpha}{q}}[\T(\sqrt{-\Delta}\gamma)]^{\alpha}+(\mathcal{K}_{\alpha,q}^{(N)})^{\alpha}\int_{\mathbb{R}^3}\big((\eta_n^2\rho_n)*{|x|^{-\alpha}}\big)\eta_n^2\rho_n dx+o_n(1)\\
&\ge [\T(\gamma^q)]^{\frac{2-\alpha}{q}}[\T(\sqrt{-\Delta}\gamma)]^{\alpha}+1-(\mathcal{K}_{\alpha,q}^{(N)})^{\alpha}\int_{\mathbb{R}^3}\big((\chi_n^2\rho_n)*{|x|^{-\alpha}}\big)\chi_n^2\rho_n dx+o_n(1).\\
\end{aligned}
\end{equation*}
Here,  we have used the definition  of $\mathcal{K}_{\alpha,q}^{(N)}$  in the second inequality,  and \eqref{eq2.8} is used in the last inequality.
Rearranging the above inequality and letting $n\to \infty$, we get
\begin{equation*}
(\mathcal{K}_{\alpha,q}^{(N)})^{\alpha} \ge\frac{\|\gamma\|_{\mathcal{S}^q}^{2-\alpha}\T(\sqrt{-\Delta}\gamma)^{\alpha}}{\int_{\mathbb{R}^3}(\rho_\gamma*|x|^{-\alpha})\rho_\gamma dx}.
\end{equation*}
This indicates that  $\gamma\not=0$ is an optimizer of (\ref{1.1}).

Case (b) $1\le\alpha<2$ and $q\in[1,\infty]$. We infer from $\|\gamma_n\|_{\mathcal{S}^q}=1$ that $\|\gamma\|_{\mathcal{S}^q}\le 1$ and $\|\eta_n\gamma_n\eta_n\|_{\mathcal{S}^q}\le1$. Thus, we obtain
\begin{equation*}
\begin{aligned}
1=\Big(\text{Tr}(\sqrt{-\Delta}\gamma_n)\Big)^\alpha &\ge \Big(\text{Tr}(\sqrt{-\Delta}\gamma)\Big)^\alpha+\Big(\text{Tr}(\sqrt{-\Delta}\eta_n\gamma_n\eta_n)\Big)^\alpha+o_n(1)\\
&\ge \|\gamma\|_{\mathcal{S}^q}^{2-\alpha}\Big(\text{Tr}(\sqrt{-\Delta}\gamma)\Big)^\alpha+\|\eta_n\gamma_n\eta_n\|_{\mathcal{S}^q}^{2-\alpha}\Big(\text{Tr}(\sqrt{-\Delta}\eta_n\gamma_n\eta_n)\Big)^\alpha+o_n(1)\\
&\ge \|\gamma\|_{\mathcal{S}^q}^{2-\alpha}\Big(\text{Tr}(\sqrt{-\Delta}\gamma)\Big)^\alpha+\mathcal{K}_{\alpha,q}^{(N)}\int_{\mathbb{R}^3}\Big((\eta_n^2\rho_n)*{|x|^{-\alpha}}\Big)\eta_n^2\rho_n dx+o_n(1)\\
&\ge \|\gamma\|_{\mathcal{S}^q}^{2-\alpha}\Big(\text{Tr}(\sqrt{-\Delta}\gamma)\Big)^\alpha+1-\mathcal{K}_{\alpha,q}^{(N)}\int_{\mathbb{R}^3}\Big((\chi_n^2\rho_n)*{|x|^{-\alpha}}\Big)\chi_n^2\rho_n dx+o_n(1).\\
\end{aligned}
\end{equation*}
This indicates  $\gamma\not=0$ is an optimizer of (\ref{1.1}).

{\textbf{(i).} For $1\le q< \frac{2-\alpha}{(1-\alpha)_+}$, we prove that the minimizing sequence satisfying \eqref{eq2.0} is sequentially   compact. Specifically, let $1\leq R:=\Rank \gamma\leq N$,  we intend to show that   $\|\gamma_n-\gamma\|_{\mathcal{S}^q}\overset{n}{\to}0$,  and by  rearranging the order of $\{i\}$ if necessary, there holds that, up to a subsequence, 
\begin{equation}\label{eq-con}
    a_i:=\lim_{n\to\infty}a_{in}>0 \text{ and } u_{in}\overset{n}{\to}u_i \text{ strongly in } H^\frac12(\mathbb{R}^3) \text{ for  }1\le i\le R;\  a_i:=\lim_{n\to\infty}a_{in}=0 \text{ for }R+1\le i\le N
\end{equation}
and 
\begin{equation}\label{eq-gamma}
    \gamma=\sum_{i=1}^R a_i |u_i\rangle \langle u_i|,\ \ \text{where}\ (u_i,u_j)_{L^2}=\delta
_{ij} \ \text{ for }\ i,j=1,2,\cdots, R.
\end{equation}

From \eqref{eq-trace1}, (\ref{eq2.4}) and (\ref{eq2.5}), we then use H\"older inequality to get that 
\begin{equation}\label{eq-dic}
\begin{aligned}
1\ge& \Big(\T(\gamma^q)+\T((\eta_n\gamma_n\eta_n)^q)\Big)^{\frac{2-\alpha}{q}}\Big(\T(\sqrt{-\Delta}\gamma)+\T(\sqrt{-\Delta}\eta_n\gamma_n\eta_n)\Big)^{\alpha}+o_n(1)\\
=&\Big[\Big(\T(\gamma^q)+\T((\eta_n\gamma_n\eta_n)^q)\Big)^{\frac{2-\alpha}{2-\alpha+q\alpha}}\Big(\T(\sqrt{-\Delta}\gamma)+\T(\sqrt{-\Delta}\eta_n\gamma_n\eta_n)\Big)^{\frac{q\alpha}{2-\alpha+q\alpha}}\Big]^{\frac{2-\alpha+q\alpha}{q}}+o_n(1)\\
\geq&\Big\{\big[\T(\gamma^q)\big]^{\frac{2-\alpha}{2-\alpha+q\alpha}}\big[\T(\sqrt{-\Delta}\gamma)\big]^{\frac{q\alpha}{2-\alpha+q\alpha}}+\big[\T((\eta_n\gamma_n\eta_n)^q)\big]^{\frac{2-\alpha}{2-\alpha+q\alpha}}\big[\T(\sqrt{-\Delta}\eta_n\gamma_n\eta_n)\big]^{\frac{q\alpha}{2-\alpha+q\alpha}}\Big\}^{\frac{2-\alpha+q\alpha}{q}}+o_n(1).
\end{aligned}
\end{equation}

We claim that 
\begin{equation}\label{eq-claim1}
    \lim_{n\to\infty}\|\eta_n\gamma_n\eta_n\|_{\mathcal{S}^q}=0.
\end{equation}
 For otherwise, assume that 
\begin{equation}\label{eq-van1}
    \liminf_{n\to\infty}\|\eta_n\gamma_n\eta_n\|_{\mathcal{S}^q}>C_1>0.
\end{equation}

If $\T(\sqrt{-\Delta}\eta_n\gamma_n\eta_n)\overset{n}\to0$, we then deduce from \eqref{1.1} that 
$$\int_{\mathbb{R}^3}\int_{\mathbb{R}^3}\frac{\eta_n^2\rho_n(x)\eta_n^2\rho_n(y)}{|x-y|^{\alpha}}dxdy\overset{n}\to0.$$
It then follows from \eqref{2.00} and \eqref{eq2.8} that 
\begin{equation*}
    \lim_{n\to\infty}\int_{\mathbb{R}^3}\big(\rho_n*|x|^{-\alpha}\big)\rho_ndx=\int_{\mathbb{R}^3}\big(\rho_\gamma*|x|^{-\alpha}\big)\rho_\gamma dx=\big(\mathcal{K}_{\alpha,q}^{(N)}\big)^{-\alpha}.
\end{equation*}
Taking \eqref{eq-van1} into the first inequality in \eqref{eq-dic}, and using the GNS-inequality \eqref{1.1}, we get that, for some $C_2>0$,
\begin{equation*}
\begin{aligned}
    1&\geq \Big(\T(\gamma^q)+C_1^q\Big)^{\frac{2-\alpha}{q}}\big[\T(\sqrt{-\Delta}\gamma)\big]^{\alpha}\ge \big[\T(\gamma^q)\big]^{\frac{2-\alpha}{q}}\big[\T(\sqrt{-\Delta}\gamma)\big]^{\alpha}+C_2\\
    &\geq (\mathcal{K}_{\alpha,q}^{(N)})^{\alpha}\int_{\mathbb{R}^3}\big(\rho_\gamma*|x|^{-\alpha}\big)\rho_\gamma dx+C_2>1.
\end{aligned} 
\end{equation*}
This leads to a contradiction. 

If $\T(\sqrt{-\Delta}\eta_n\gamma_n\eta_n)>C_3>0$ for some $C_3>0$.  Since  $\frac{2-\alpha+q\alpha}{q}>1$, we  deduce from \eqref{eq-dic} and \eqref{eq-van1} that 
\begin{equation*}
\begin{aligned}
1\ge& [\T(\gamma^q)]^{\frac{2-\alpha}{q}}[\T(\sqrt{-\Delta}\gamma]^{\alpha}+[\T((\eta_n\gamma_n\eta_n)^q)]^{\frac{2-\alpha}{q}}[\T(\sqrt{-\Delta}\eta_n\gamma_n\eta_n)]^\alpha+o_n(1)\\
&+\frac{2-\alpha+q\alpha}{q}[\T(\gamma^q)]^{\frac{2-\alpha}{q}-\frac{2-\alpha}{2-\alpha+q\alpha}}[\T(\sqrt{-\Delta}\gamma)]^{\alpha-\frac{q\alpha}{2-\alpha+q\alpha}}[\T((\eta_n\gamma_n\eta_n)^q)]^{\frac{2-\alpha}{2-\alpha+q\alpha}}[\T(\sqrt{-\Delta}\eta_n\gamma_n\eta_n)]^{\frac{q\alpha}{2-\alpha+q\alpha}}\\
\ge& (\mathcal{K}_{\alpha,q}^{(N)})^{\alpha}\int_{\mathbb{R}^3}\big(\rho_\gamma*|x|^{-\alpha}\big)\rho_\gamma dx+(\mathcal{K}_{\alpha,q}^{(N)})^{\alpha}\int_{\mathbb{R}^3}\big[(\eta_n^2\rho_n)*{|x|^{-\alpha}}\big]\eta_n^2\rho_n dx+o_n(1)\\
&+\frac{2-\alpha+q\alpha}{q}[\T(\gamma^q)]^{\frac{2-\alpha}{q}-\frac{2-\alpha}{2-\alpha+q\alpha}}[\T(\sqrt{-\Delta}\gamma)]^{\alpha-\frac{q\alpha}{2-\alpha+q\alpha}}[\T((\eta_n\gamma_n\eta_n)^q)]^{\frac{2-\alpha}{2-\alpha+q\alpha}}[\T(\sqrt{-\Delta}\eta_n\gamma_n\eta_n)]^{\frac{q\alpha}{2-\alpha+q\alpha}}\\
>&1,
\end{aligned}
\end{equation*}
where we have used \eqref{2.00} and \eqref{eq2.8} in the last inequality. This also leads to a contradiction. Therefore, claim \eqref{eq-claim1} is obtained. This implies from GNS-inequality
that
\begin{equation*}
\int_{\mathbb{R}^3}\big[(\eta_n^2\rho_n)*{|x|^{-\alpha}}\big]\eta_n^2\rho_n dx=o_n(1)
\end{equation*}
and
\begin{equation*}
\int_{\mathbb{R}^3}\big(\rho_n*{|x|^{-\alpha}}\big)\rho_n dx=\int_{\mathbb{R}^3}\big(\rho_\gamma*|x|^{-\alpha}\big)\rho_\gamma dx+o_n(1).
\end{equation*}
As a consequence, we have  
\begin{equation*}
\begin{aligned}
    1=&\|\gamma_n\|_{\mathcal{S}^q}^{2-\alpha}(\T(\sqrt{-\Delta}\gamma_n))^{\alpha}\ge \|\gamma\|_{\mathcal{S}^q}^{2-\alpha}(\T(\sqrt{-\Delta}\gamma))^{\alpha}+o_n(1)\\
    \ge& (\mathcal{K}_{\alpha,q}^{(N)})^{\alpha}\int_{\mathbb{R}^3}\big(\rho_\gamma*|x|^{-\alpha}\big)\rho_\gamma dx+o_n(1)=1+o_n(1)
\end{aligned}
\end{equation*}
and
\begin{equation}\label{eq2.16}
1=\T(\gamma_n^q)=\T(\gamma^q)+o_n(1)\ \text{and}\ 1=\T(\sqrt{-\Delta}\gamma_n)=\T(\sqrt{-\Delta}\gamma)+o_n(1).
\end{equation}
It follows from Theorem 2.16 in \cite{s} that $\|\gamma_n-\gamma\|_{\mathcal{S}^q}\overset{n}\to 0$. In addition, from \eqref{eq2.000} and \eqref{eq2.0} , we  see that, up to subsequence, there holds 
\begin{equation}\label{eq2.17}
    a_i:=\lim_{n\to\infty}a_{in}\geq0 \text{ and } u_{in}\overset{n}{\rightharpoonup }u_i \text{ in } H^\frac12(\mathbb{R}^3).
\end{equation}
Thus, $\gamma^q$ can be expressed as 
\begin{equation}\label{eq2.18}
    \gamma^q=\sum_{i=1}^N a_i^q |u_i\rangle \langle u_i|.
\end{equation}
We deduce from \eqref{eq2.0}, {
\eqref{eq2.16} and \eqref{eq2.17} that 
\begin{equation*}
\begin{aligned}
1=\T(\gamma_n^q)=\sum_{i=1}^N a_{in}^q\stackrel{n}\to\sum_{i=1}^N a_i^q=\T(\gamma^q)
\leq \sum_{i=1}^N a_i^q \langle u_i|u_i\rangle \le 1.
\end{aligned}
\end{equation*}
This implies that
\begin{equation}\label{eq2.20}
    (u_i,u_i)_{L^2}=1 \text{ provided }a_i>0, \text{ for }\ i=1,2,\cdots,N. 
\end{equation}
Recall $\Rank \gamma=\Rank \gamma^q=R$, rearranging the order of $\{i\}$ if necessary, we   obtain  that
\begin{equation}\label{eq2.19}
    a_i>0 \ \text{ and }\ u_i\not=0 \ \text{ for }\ i=1,2,\cdots,R.
\end{equation}
This combine with \eqref{eq2.18} and \eqref{eq2.20} gives that  $(u_i,u_i)_{L^2}=1$ for $i=1,2,\cdots,R$. This further implies 
\begin{equation}\label{eq2.21}
    u_{in}\overset{n}\to u_i \text{ in } L^2(\mathbb{R}^3;\mathbb{C}),\ \text{for all }i=1,2,\cdots,R.
\end{equation}
Recall that $(u_{in},u_{jn})_{L^2}=\delta_{ij}$ for all $i,j=1,2,\cdots,N,$  and $n\in\mathbb{N}^+$, we thus deduce that
\begin{equation}\label{eq2.22}
    (u_i,u_j)_{L^2}=\delta
_{ij} \ \text{ for }\ i,j=1,2,\cdots, R.
\end{equation}
We next prove that if $R<N$, then 
\begin{equation}\label{eq2.23}
    a_i=0 \ \text{for $i=R+1,\cdots, N$}.
\end{equation}
For otherwise, if $a_{R+1}>0$, then $u_{R+1}\not=0$ followed by \eqref{eq2.20}.  This together with \eqref{eq2.21} indicates that 
 $(u_i,u_{R+1})_{L^2}=0$ for $i=1,2,\cdots, R$. It then follows from \eqref{eq2.18} and \eqref{eq2.22} that $\Rank \gamma\geq R+1$, which however contradicts that $\Rank \gamma^q=R$. 
 \eqref{eq2.23} is thus proved. From \eqref{eq2.21} to \eqref{eq2.23} we see that  \eqref{eq-gamma} holds.  Moreover,  we obtain by noting from \eqref{2.00} that  
It follows  from \eqref{2.00} that
\begin{equation*}
\int_{\mathbb{R}^3}\Big(\rho_n*{|x|^{-\alpha}}\Big)\rho_n dx\overset{n}\to \int_{\mathbb{R}^3}\Big(\rho_\gamma*{|x|^{-\alpha}}\Big)\rho_\gamma dx={\big(\mathcal{K}_{\alpha,q}^{(N)}\big)^{-\alpha}}.
\end{equation*}
Therefore, 
\begin{equation*}
1=[\T (\sqrt{-\Delta}\gamma_n)]^\alpha\geq[\T (\sqrt{-\Delta}\gamma)]^\alpha \ge \big(\mathcal{K}_{\alpha,q}^{(N)}\big)^{\alpha} \frac{\int_{\mathbb{R}^3}(\rho_\gamma*|x|^{-\alpha})\rho_\gamma dx}{[\T(\gamma^q)]^{2-\alpha} }= 1.
\end{equation*}
This together with \eqref{eq2.21} and \eqref{eq2.23} indicates that $ u_{in} \overset{n}\to  u_i $ in $H^\frac{1}{2}(\mathbb{R}^3)$  for $1\leq i\leq R$.  Then, \eqref{eq-con} follows by recalling \eqref{eq2.19} and \eqref{eq2.23}. 

\textbf{II. Explicit expression of optimizers.} In this part, we will reveal  the explicit expression of the optimizer $\gamma$. Normalize an optimizer $\gamma$ such that
\begin{equation}\label{eq2.025}
\text{Tr}(\sqrt{-\Delta}\gamma)=\int_{\mathbb{R}^3}(\rho_\gamma*|x|^{-\alpha})\rho_\gamma dx=1 \ \text{ and }\ \|\gamma\|_{\mathcal{S}^q}=(\mathcal{K}_{\alpha,q}^{(N)})^\frac{\alpha}{2-\alpha}.
\end{equation}
 The proof  is divided into two cases:  case (A):  $0<\alpha<2$ and $1\le q<\frac{2-\alpha}{(1-\alpha)_+}(<\infty)$,  and case (B): $1\le \alpha<2$ and $q=\infty$.

{\em Case (A)   $0<\alpha<2$ and $1\le q<\frac{2-\alpha}{(1-\alpha)_+}$}. 
Choosing a smooth curve of operators
\begin{equation*}
0\leq \gamma^*(t)=\gamma(t)=\gamma+t\delta+o(t), \quad {\rm rank~}\gamma(t)\le N
\end{equation*}
and substituting $\gamma(t)$ into (\ref{1.1}), we get
\begin{equation}\label{2.4}
\begin{aligned}
(\mathcal{K}_{\alpha,q}^{(N)})^{\alpha}&\le \frac{\big[\text{Tr}(\gamma^q(t))\big]^{\frac{2-\alpha}{q}}[\T(\sqrt{-\Delta}\gamma(t))]^{\alpha}}{\int_{\mathbb{R}^3}\big(\rho_{\gamma(t)}*|x|^{-\alpha}\big)\rho_{\gamma(t)}dx}\\
&\le \frac{\big[\T(\gamma^q)+qt\text{Tr}(\delta\gamma^{q-1})+o(t)\big]^{\frac{2-\alpha}{q}}[\text{Tr}(\sqrt{-\Delta}\gamma)+t\text{Tr}(\sqrt{-\Delta}\delta)+o(t)]^\alpha}{\int_{\mathbb{R}^3}\int_{\mathbb{R}^3}[\rho_\gamma+t\rho_\delta+o(t)](x)[\rho_\gamma+t\rho_\delta+o(t)](y)|x-y|^{-\alpha}dxdy}\\
&= [\T(\gamma^q)]^{\frac{2-\alpha}{q}}\frac{[1+qt\frac{\text{Tr}(\delta\gamma^{q-1})}{\T(\gamma^q)}+o(t)]^{\frac{2-\alpha}{q}}[1+t\T(\sqrt{-\Delta}\delta)+o(t)]^\alpha}{1+2t\int_{\mathbb{R}^3}\int_{\mathbb{R}^3}\rho_{\gamma}(x)\rho_{\delta}(y)|x-y|^{-\alpha}dxdy+o(t)}\\
&\le (\mathcal{K}_{\alpha,q}^{(N)})^{\alpha}\Big[1+t\T\Big(\delta\big[\alpha\sqrt{-\Delta}-2\rho_\gamma*|x|^{-\alpha}+(2-\alpha)\frac{\gamma^{q-1}}{\text{Tr}(\gamma^q)}\big]\Big)+o(t)\Big],
\end{aligned}
\end{equation}
where in the denominator of the second inequality $[\rho_\gamma+t\rho_\delta+o(t)](x):=\rho_\gamma(x)+t\rho_\delta(x)+o(t)(x)$ and here $o(t)$ in integral is a function with respect to $x$ or $y$. It follows
\begin{equation}\label{eq-E}
1\le 1+t\text{Tr}\Big[\delta\Big(\alpha\sqrt{-\Delta}-2\rho_\gamma*|x|^{-\alpha}+(2-\alpha)\frac{\gamma^{q-1}}{\text{Tr}(\gamma^q)}\Big)\Big]+o(t),\quad \text{for $|t|$ small\ enough}.
\end{equation}
We deduce by taking $\gamma(t):=e^{{\rm i}tH}\gamma e^{-{\rm i}tH}=\gamma+{ \text{i}}t[H,\gamma]+o(t)$, for some self-adjoint (smooth and finite rank) operator $H$,
\begin{equation*}
\begin{aligned}
0&=\text{Tr}\Big([H,\gamma]\Big(\alpha\sqrt{-\Delta}-2\rho_\gamma*|x|^{-\alpha}+(2-\alpha)\frac{\gamma^{q-1}}{\text{Tr}(\gamma^q)}\Big)\Big)\\
&=\text{Tr}\Big(H\Big[\gamma,\alpha\sqrt{-\Delta}-2\rho_\gamma*|x|^{-\alpha}+(2-\alpha)\frac{\gamma^{q-1}}{\text{Tr}(\gamma^q)}\Big]\Big),
\end{aligned}
\end{equation*}
where $[\cdot,\cdot]$ is Lie bracket and $[a,b]=ab-ba$. It is apparent from the arbitrariness of $H$ that
\begin{equation*}
\Big[\gamma,\alpha\sqrt{-\Delta}-2\rho_\gamma*|x|^{-\alpha}+(2-\alpha)\frac{\gamma^{q-1}}{\text{Tr}(\gamma^q)}\Big]=0.
\end{equation*}
Therefore, we obtain the explicit expression of $\gamma$, which formed by  some eigenfunctions of $H_{\alpha,\gamma}:=\sqrt{-\Delta}-\frac{2}{\alpha} \rho_\gamma*|x|^{-\alpha}$,
\begin{equation*}
\gamma=\sum_{i=1}^R k_i|u_{n_i}\rangle \langle u_{n_i}|, \quad \text{for some}\ u_{n_i}\ \text{ satisfying } \ H_{\alpha,\gamma} u_{n_i}=\mu_{n_i}u_{n_i},
\end{equation*}
assuming $H_{\alpha,\gamma}$ admits at least $R$ eigenvalues, $1\le R\le N$.

It follows from selecting $\delta$ supported on the range of $\gamma$ in \eqref{eq-E} that, 
\begin{equation*}
\sqrt{-\Delta}u_{n_i}-\frac{2}{\alpha}\rho_\gamma*|x|^{-\alpha}u_{n_i}+\frac{2-\alpha}{\alpha}(\text{Tr}(\gamma^q))^{-1}\gamma^{q-1}u_{n_i}\equiv 0.
\end{equation*}
We infer that
\begin{equation*}
\mu_{n_i}=-\frac{2-\alpha}{\alpha}\frac{k_i^{q-1}}{\T(\gamma^q)}<0 \text{\  and thus \ }k_i=\big[\frac{\alpha}{2-\alpha}\T\big(\gamma^q)\big]^{\frac{1}{q-1}}|\mu_{n_i}|^{\frac{1}{q-1}}.
\end{equation*}
Then plugging the expression of $k_i$ into the identity $\T(\gamma^q)=\sum\limits_{i=1}^Rk_i^q$, we deduce that 
$$\big(\T(\gamma^q)\big)^\frac{1}{q-1}=\big(\frac{2-\alpha}{\alpha}\big)^\frac{q}{q-1}\Big(\sum\limits_{k=1}^R|\mu_{n_k}|^{\frac{q}{q-1}}\Big)^{-1}.$$ Eventually,  
 for $0<\alpha<2$ and $1< q<\frac{2-\alpha}{(1-\alpha)_+}$, we get that
\begin{equation}\label{iii3}
\gamma=\sum_{i=1}^R k_i|u_{n_i}\rangle \langle u_{n_i}|,\quad \text{ where }\quad
k_i=\frac{2-\alpha}{\alpha}\Big(\sum_{k=1}^R|\mu_{n_k}|^{\frac{q}{q-1}}\Big)^{-1}|\mu_{n_i}|^{\frac{1}{q-1}}.
\end{equation}

{\em Case (B)  $1\le \alpha<2$ and $q=\infty$}.  
Similar to Case (A), we  take a smooth curve of operators $\gamma(t)$
\begin{equation*}
0\leq\gamma^*(t)=\gamma(t)=\gamma+t\delta+o(t), \quad {\rm rank~}\gamma(t)\le N,
\end{equation*}
 where $\delta$ is chosen  such that $\|\gamma(t)\|\le \|\gamma\|$.  For $|t|$ small enough,
\begin{equation*}
\begin{aligned}
(\mathcal{K}_{\alpha,\infty}^{(N)})^{\alpha}&\le \frac{\|\gamma(t)\|^{2-\alpha}[\T(\sqrt{-\Delta}\gamma(t))]^\alpha}{\int_{\mathbb{R}^3}\big(\rho_{\gamma(t)}*|x|^{-\alpha}\big)\rho_{\gamma(t)}dx}
\le (\mathcal{K}_{\alpha,\infty}^{(N)})^{\alpha}\frac{1+\alpha t\T(\sqrt{-\Delta}\delta)+o(t)}{1+2t\int_{\mathbb{R}^3}(\rho_\gamma*|x|^{-\alpha})\rho_\delta dx+o(t)}\\
&\le (\mathcal{K}_{\alpha,\infty}^{(N)})^{\alpha}\big\{1+t\T[\delta(\alpha\sqrt{-\Delta}-2\rho_\gamma*|x|^{-\alpha})]+o(t)\big\}.
\end{aligned}
\end{equation*}
Taking $\delta={\rm i}[H,\gamma]$ for some self-adjoint and finite rank operator $H$, then we obtain that
\begin{equation*}
\Big[\gamma,\ \alpha \sqrt{-\Delta}-2\rho_\gamma*|x|^{-\alpha}\Big]=0.
\end{equation*}
It follows that $\gamma$ can be expressed by eigenfunctions of $H_{\alpha,\gamma}$, i.e., 
\begin{equation*}
\gamma=\sum_{i=1}^R k_i|u_{n_i}\rangle \langle u_{n_i}|, \quad \text{for some}\ u_{n_i}\ \text{ satisfying } \ H_{\alpha,\gamma} u_{n_i}=\mu_{n_i}u_{n_i}.
\end{equation*}

We \textbf{claim} that for all $1\leq i\le R$,
\begin{equation}\label{eq2.28}
    \text{$\mu_{n_i}\leq0$ if $1<\alpha<2$ , and  $\mu_{n_i}<0$ if $\alpha=1$}.
\end{equation}
 Actually, let
\begin{equation*}
\gamma'=\gamma-\varepsilon |u_{n_i}\rangle \langle u_{n_i}|,
\end{equation*}
where $0<\varepsilon <k_i$ is small enough. Then, $\|\gamma'\|\leq \|\gamma\|=(\mathcal{K}_{\alpha,q}^{(N)})^\frac{\alpha}{2-\alpha}$, $\rho_{\gamma'}=\rho_\gamma-\varepsilon |u_{n_i}|^2$ and 
\begin{equation}\label{eq2.29}
\T(\sqrt\Delta\gamma')=\T(\sqrt{\Delta}\gamma)-\varepsilon\int_{\mathbb{R}^3}|(-\Delta)^\frac{1}{4}u_{n_i}|^2dx=1-\varepsilon\int_{\mathbb{R}^3}|(-\Delta)^\frac{1}{4}u_{n_i}|^2dx.
\end{equation}
Moreover, we can deduce from $H_{\alpha,\gamma} u_{n_i}=\mu_{n_i}u_{n_i}$ that 
$$\int_{\mathbb{R}^3}|(-\Delta)^\frac{1}{4}u_{n_i}|^2dx=\mu_{n_i}+\frac{2}{\alpha}\int_{\mathbb{R}^3}(\rho_{\gamma}*|x|^{-\alpha})|u_{n_i}|^2dx.$$
 Recalling \eqref{eq2.025} and plugging the above estimates into the inequality $$(\mathcal{K}_{\alpha,\infty}^{(N)})^{\alpha}\le \frac{\|\gamma'\|^{2-\alpha}\big(\T(\sqrt{-\Delta}\gamma')\big)^{\alpha}}{\int_{\mathbb{R}^3}(\rho_{\gamma'}*|x|^{-\alpha})\rho_{\gamma'}dx},$$ we arrive at  
 \begin{equation}\label{eq2.30}
     1-2\varepsilon\int_{\mathbb{R}^3}(\rho_{\gamma}*|x|^{-\alpha})|u_{n_i}|^2dx+\varepsilon^2\int_{\mathbb{R}^3}(|u_{n_i}|^2*|x|^{-\alpha})|u_{n_i}|^2dx\le \big(1-\varepsilon \mu_{n_i}-\frac{2\varepsilon}{\alpha}\int_{\mathbb{R}^3}(\rho_{\gamma}*|x|^{-\alpha})|u_{n_i}|^2dx\big)^\alpha.
 \end{equation}
We recall  that for $|x|<1$, $(1-x)^\alpha\leq 1-\alpha x$ if $\alpha\in(0,1]$, and $(1-x)^\alpha\leq 1-\alpha x+o(x)$ if $\alpha\in(1,2)$. Applying this inequality to the RHS of \eqref{eq2.30}, a simple calculation then gives that 
 \begin{equation*}
      \alpha\mu_{n_i}\leq -\varepsilon\int_{\mathbb{R}^3}(|u_{n_i}|^2*|x|^{-\alpha})|u_{n_i}|^2dx \text{\ for \ } \alpha=1,
 \end{equation*}
 and 
 \begin{equation*}
      \alpha\mu_{n_i}\leq -\varepsilon\int_{\mathbb{R}^3}(|u_{n_i}|^2*|x|^{-\alpha})|u_{n_i}|^2dx+O(\varepsilon) \text{\ for \ } \alpha\in(1,2).
 \end{equation*}
 These two estimates  indicate  that claim \eqref{eq2.28} holds.

We next prove  that when $1\le \alpha<2$ and $q=\infty$, then,
\begin{equation}\label{eq2.31}
    k_i\equiv k_j:=k>0 \text{\  for all \ } i,j \in\{1\leq l\leq R: \mu_{n_l}<0\}.
\end{equation}
We argue by contradiction.  Assume that there exists $\mu_{n_m}<0$ and $k_m<\|\gamma\|$,   then we choose $t>0$ small enough such that $t+k_m\le \|\gamma\|$, and construct $\gamma(t)=\gamma+t|u_m\rangle \langle u_m|$. Similar to the arguments of \eqref{eq2.29} and \eqref{eq2.30}, one can derive the following contradiction:
\begin{equation*}
\begin{aligned}
(\mathcal{K}_{\alpha,\infty}^{(N)})^{\alpha}&\le \frac{\|\gamma(t)\|^{2-\alpha}(\T(\sqrt{-\Delta}\gamma(t)))^{\alpha}}{\int_{\mathbb{R}^3}(\rho_{\gamma(t)}*|x|^{-\alpha})\rho_{\gamma(t)}dx}=(\mathcal{K}_{\alpha,\infty}^{(N)})^{\alpha}\frac{1+t\alpha\mu_{m}+2t\int_{\mathbb{R}^3}(\rho_\gamma*|x|^{-\alpha})|u_m|^2dx+o(t)}{\int_{\mathbb{R}^3}\big[(\rho_\gamma+t|u_m|^2)*|x|^{-\alpha}\big](\rho_\gamma+t|u_m|^2)dx}\\
&\le (\mathcal{K}_{\alpha,\infty}^{(N)})^{\alpha}(1+t\alpha\mu_m+o(t))<(\mathcal{K}_{\alpha,\infty}^{(N)})^{\alpha}.
\end{aligned}
\end{equation*}
It is a contradiction, and \eqref{eq2.31} is obtained.

As the end of  (ii), we  show  that $H_{\alpha, \gamma}$ has at most $R$ negative eigenvalue provided  $R<N$. If not, there exists $R+1\le N$ and $\mu_{n_{R+1}}<0$, then we consider the operator
\begin{equation*}
\gamma(t)=\gamma+t|u_{n_{R+1}}\rangle \langle u_{n_{R+1}}|,\quad \text{for $t> 0$ small enough}.
\end{equation*}
A routine computation similar to (\ref{2.4}) leads to the following contradiction
\begin{equation*}
\begin{aligned}
(\mathcal{K}_{\alpha,q}^{(N)})^{\alpha} &\le (\mathcal{K}_{\alpha,q}^{(N)})^{\alpha}\Big(1+t\Big(u_{n_{R+1}},\big(\alpha\sqrt{-\Delta}-2\rho_\gamma*|x|^{-\alpha}\big)u_{n_{R+1}}\Big)+o(t)\Big)\\
&\le (\mathcal{K}_{\alpha,q}^{(N)})^{\alpha}(1+t\mu_{R+1}+o(t))<(\mathcal{K}_{\alpha,q}^{(N)})^{\alpha}.
\end{aligned}
\end{equation*}
Thus, $H_{\alpha,\gamma}$ has at most $R$ negative eigenvalues when $R<N$.

\textbf{(iii).} The explicit expression of $k_i$ can be obtained by   (\ref{iii3}) for the case of $0<\alpha<2$ and $1< q<\frac{2-\alpha}{(1-\alpha)_+}$}, and  by \eqref{eq2.28} and \eqref{eq2.31} for the case of $\alpha=1$ and $q=\infty$, respectively.  It remains to  prove that  $\{\mu_{n_i}\}_{1\le i\le R}$ are $R$ first negative eigenfunctions of $H_{\alpha,\gamma}$ provided 
either $0<\alpha\le 1$, or $1<\alpha<2$ and  $\Rank \gamma<N$.
  We argue by contradiction and assume that the $i$-th eigenfunction of $H_{\alpha,\gamma}$, denoted by $u_i$, corresponding to $\mu_i<\mu_{n_i}<0$, does not belong to the range of $\gamma$.
  
  When $0<\alpha\le 1$,
  we construct an operator
\begin{equation*}
\gamma':=\gamma-k_i|u_{n_i}\rangle \langle u_{n_i}|+k_i|u_i\rangle \langle u_i|:=\gamma+\delta.
\end{equation*}
Then, $\|\gamma\|_{\mathcal{S}^q}=\|\gamma'\|_{\mathcal{S}^q}$, and it follows that  
\begin{equation*}
\begin{aligned}
0<\text{Tr}(\sqrt{-\Delta}\gamma')&=\text{Tr}(\sqrt{-\Delta}\gamma)+k_i\langle u_i,\sqrt{-\Delta}u_i\rangle-k_i\langle u_{n_i},\sqrt{-\Delta}u_{n_i}\rangle\\
&=1+\frac{2k_i}{\alpha}\int_{\mathbb{R}^3}(\rho_\gamma*|x|^{-\alpha})(|u_i|^2-|u_{n_i}|^2)+(\mu_i-\mu_{n_i})k_i\\
&<1+\frac{2k_i}{\alpha}\int_{\mathbb{R}^3}(\rho_\gamma*|x|^{-\alpha})(|u_i|^2-|u_{n_i}|^2).
\end{aligned}
\end{equation*}
 It also follows from the non-negativity properties of the Hartree-type convolution \cite{ll} that
\begin{equation*}
\begin{aligned}
\int_{\mathbb{R}^3}(\rho_{\gamma'}*|x|^{-\alpha})\rho_{\gamma'}dx=&\int_{\mathbb{R}^3}(\rho_{\gamma}*|x|^{-\alpha})\rho_{\gamma}dx+2k_i\int_{\mathbb{R}^3}(\rho_{\gamma}*|x|^{-\alpha})(|u_i|^2-|u_{n_i}|^2)dx\\
&+k_i^2\int_{\mathbb{R}^3}\big[(|u_i|^2-|u_{n_i}|^2)*|x|^{-\alpha}\big](|u_i|^2-|u_{n_i}|^2)dx\\
\ge&1+2k_i\int_{\mathbb{R}^3}(\rho_{\gamma}*|x|^{-\alpha})(|u_i|^2-|u_{n_i}|^2)dx>1-\alpha\ge0.
\end{aligned}
\end{equation*}

Since $\alpha\leq1$, we see that
\begin{equation*}
\begin{aligned}
(\mathcal{K}_{\alpha,q}^{(N)})^{\alpha}\leq \frac{\|\gamma'\|_{\mathcal{S}^q}^{2-\alpha}(\T(\sqrt{-\Delta}\gamma'))^\alpha}{\int_{\mathbb{R}^3}\rho_{\gamma'}*|x|^{-\alpha}\rho_{\gamma'}dx}&<(\mathcal{K}_{\alpha,q}^{(N)})^{\alpha}\frac{\Big(1+\frac{2k_i}{\alpha}\int_{\mathbb{R}^3}(\rho_\gamma*|x|^{-\alpha})(|u_i|^2-|u_{n_i}|^2)dx\Big)^{\alpha}}{1+2k_i\int_{\mathbb{R}^3}(\rho_{\gamma}*|x|^{-\alpha})(|u_i|^2-|u_{n_i}|^2)dx}<(\mathcal{K}_{\alpha,q}^{(N)})^{\alpha}.
\end{aligned}
\end{equation*}
This leads to a contradiction, hence $\mu_{n_i}=\mu_i$, for all $i$.

When $1<\alpha<2$ and rank $\gamma<N$,  we set
\begin{equation*}
\gamma':=\gamma-\varepsilon|u_{n_i}\rangle \langle u_{n_i}|+\varepsilon |u_{i}\rangle \langle u_{i}|,
\end{equation*}
where $0<\varepsilon<k_i$ is small enough. Then, $\Rank \gamma'\leq N$ and  $\|\gamma'\|_{\mathcal{S}^q}\leq\|\gamma\|_{\mathcal{S}^q}=(\mathcal{K}_{\alpha,q}^{(N)})^\frac{2-\alpha}{\alpha}$ for $q\geq1$. Moreover, we have 
\begin{equation*}
    \int_{\mathbb{R}^3}(\rho_{\gamma'}*|x|^{-\alpha})\rho_{\gamma'}dx>1+2\varepsilon\int_{\mathbb{R}^3}(\rho_{\gamma}*|x|^{-\alpha})(|u_i|^2-|u_{n_i}|^2)dx
\end{equation*}
and
\begin{equation*}
\T (\sqrt{-\Delta}\gamma')=1+\frac{2\varepsilon}{\alpha}\int_{\mathbb{R}^3}(\rho_\gamma*|x|^{-\alpha})(|u_i|^2-|u_{n_i}|^2)+(\mu_i-\mu_{n_i})\varepsilon.
\end{equation*}
By using $\mu_i<\mu_{n_i}<0$, we still can derive the following contradiction:
\begin{equation*}
\begin{aligned}
(\mathcal{K}_{\alpha,q}^{(N)})^{\alpha}&\le\frac{\|\gamma'\|_{\mathcal{S}^q}^{2-\alpha}(\T(\sqrt{-\Delta}\gamma'))^\alpha}{\int_{\mathbb{R}^3}\rho_{\gamma'}*|x|^{-\alpha}\rho_{\gamma'}dx}<(\mathcal{K}_{\alpha,q}^{(N)})^{\alpha}\frac{\Big(1+\frac{2\varepsilon}{\alpha}\int_{\mathbb{R}^3}(\rho_\gamma*|x|^{-\alpha})(|u_i|^2-|u_{n_i}|^2)dx+(\mu_i-\mu_{n_i})\varepsilon\Big)^{\alpha}}{1+2\varepsilon\int_{\mathbb{R}^3}(\rho_{\gamma}*|x|^{-\alpha})(|u_i|^2-|u_{n_i}|^2)dx}\\
<&(\mathcal{K}_{\alpha,q}^{(N)})^{\alpha}\frac{1+2\varepsilon\int_{\mathbb{R}^3}(\rho_{\gamma}*|x|^{-\alpha})(|u_i|^2-|u_{n_i}|^2)dx
+\alpha\varepsilon (\mu_i-\mu_{n_i})+o(\varepsilon)}{1+2\varepsilon\int_{\mathbb{R}^3}(\rho_{\gamma}*|x|^{-\alpha})(|u_i|^2-|u_{n_i}|^2)dx} 
<(\mathcal{K}_{\alpha,q}^{(N)})^{\alpha}.
\end{aligned}
\end{equation*}


\textbf{III. Existence of strict decreasing subsequence of $\mathcal{K}_{\alpha,q}^{(N)}$ w.r.t. \emph{N}.} Motivated by \cite{gln}, we first claim that when $0<\alpha\le1$, $q$ is close to $\frac{2-\alpha}{(1-\alpha)_+}$ enough, then 
\begin{equation}\label{eq2.12}
\mathcal{K}_{\alpha,q}^{(2N)}<\mathcal{K}_{\alpha,q}^{(N)} \ \text{ provided $\mathcal{K}_{\alpha,q}^{(N)}$ has an optimizer $\gamma$ of rank $N$.}
\end{equation}
Indeed, let $\gamma:=\sum_{j=1}^N k_j|u_j\rangle \langle u_j|$ be a minimizer  for $\mathcal{K}_{\alpha,q}^{(N)}$ of rank $N$,  such that
\begin{equation*}
\text{Tr}(\sqrt{-\Delta} \gamma)=\int_{\mathbb{R}^3}(\rho_\gamma*|x|^{-\alpha})\rho_\gamma dx=1
\end{equation*}
and $(\mathcal{K}_{\alpha,q}^{(N)})^{\frac{\alpha}{2-\alpha}}=(\T (\gamma^q))^{\frac{1}{q}}$, if $q\not =\infty$, $\mathcal{K}_{1,\infty}^{(N)}=\|\gamma\|$, if $q=\infty$.
Furthermore, there exist $N$ functions $u_1,\ u_2,...,u_N$ satisfy
\begin{equation*}
\sqrt{-\Delta} u_j-\frac{2}{\alpha}(\rho_\gamma*|x|^{-\alpha})u_j=\mu_j u_j,\quad \forall\ 1\le j\le N.
\end{equation*}
Proceeding as in the proof of Theorem \ref{A.1} in Appendix, we have
\begin{equation}\label{2.5}
|u_j|\le \frac{C}{1+|x|^4}, \quad \forall\ j=1,2,...,N.
\end{equation}

For $\bar R>0$, we define $u_{j,\bar R}(x):=u_j(x-\bar Re_1)$ with $e_1=(1,0,...,0)$. And we construct the Gram matrix
\begin{equation*}
G_{\bar R}= \left( \begin{array}{cc}
E_{N\times N} & E^{\bar R} \\
(E^{\bar R})^* & E_{N\times N}
\end{array}\right)
\end{equation*}
with $E_{N\times N}$ is $N\times N$ identity matrix and $E^{\bar R}=(E_{ij}^{\bar R})_{N\times N}$ where
\begin{equation*}
E_{ij}^{\bar R}:=( u_i, u_{j,\bar R})_{L^2}=\int_{\mathbb{R}^3}u_i(x)u_j(x-\bar Re_1)dx.
\end{equation*}
Now we need to estimate $|E_{ij}^{\bar R}|$. Using (\ref{2.5}), we get
\begin{equation*}
\begin{aligned}
|E_{ij}^{\bar R}|&\le \int_{\mathbb{R}^3}|u_i(x)||u_{j}(x-\bar Re_1)|dx\le  \int_{\mathbb{R}^3}\frac{C}{(1+|x|^4)(1+|x-\bar Re_1|^4)}dx.
\end{aligned}
\end{equation*}
It follows from Lemma A.3 in \cite{gnnt} that
$|E_{ij}^{\bar R}|\le C(\frac{1}{\bar R^4}+\frac{1}{\bar R^5}).$
 We  construct
\begin{equation*}
\left(\begin{array}{cc}
\psi_{1,\bar R}\\
\vdots\\
\psi_{N,\bar R}\\
\psi_{N+1, \bar R}\\
\vdots\\
\psi_{2N,\bar R}
\end{array}\right)=(G_{\bar R})^{-\frac{1}{2}}
\left(\begin{array}{cc}
u_1\\
\vdots\\
u_N\\
u_{1,\bar R}\\
\vdots\\
u_{N,\bar R}
\end{array}\right)
\end{equation*}
and
\begin{equation*}
\gamma_{\bar R}=\sum_{i=1}^Nk_i\Big(|\psi_{i,\bar R}\rangle\langle\psi_{i,\bar R}|+|\psi_{N+i,\bar R}\rangle \langle \psi_{N+i,\bar R}|\Big).
\end{equation*}
We have
\begin{equation}\label{2.6.1}
\text{Tr}(\gamma_{\bar R}^q)=2\text{Tr}(\gamma^q),\quad \|\gamma_{\bar R}\|=\|\gamma\|.
\end{equation}
Substituting $\gamma_{\bar R}$ into inequality (\ref{1.1}), we obtain
\begin{equation*}
(\mathcal{K}_{\alpha,q}^{(2N)})^{\alpha}\le (\mathcal{K}_{\alpha,q}^{(N)})^\alpha\frac{2^{\frac{2-\alpha}{q}}(\text{Tr}\sqrt{-\Delta}\gamma_{\bar R})^{\alpha}}{\int_{\mathbb{R}^3}(\rho_{\gamma_{\bar R}}*|x|^{-\alpha})\rho_{\gamma_{\bar R}}dx}
\end{equation*}
and
\begin{equation*}
\mathcal{K}_{1,\infty}^{(2N)}\le \mathcal{K}_{1,\infty}^{(N)}\frac{\text{Tr}(\sqrt{-\Delta}\gamma_{\bar R})}{\int_{\mathbb{R}^3}(\rho_{\gamma_{\bar R}}*|x|^{-1})\rho_{\gamma_{\bar R}}dx}.
\end{equation*}
Expanding the Gram matrix $G_{\bar R}$ as in \cite{gln}
\begin{equation*}
(S_{\bar R})^{-\frac{1}{2}}=\left( \begin{array}{cc}
E_{N\times N} &0 \\
0 & E_{N\times N}
\end{array}\right)-\frac{1}{2}\left( \begin{array}{cc}
0 & E^{\bar R} \\
(E^{\bar R})^* & 0
\end{array}\right)+o((\max_{i,j}E_{ij}^{\bar R})^2),
\end{equation*}
where
\begin{equation*}
\max_{i,j}E_{ij}^{\bar R}=\max_{i,j}\int_{\mathbb{R}^3}|u_i(x)||u_j(x-\bar Re_1)|dx=O(\bar R^{-4}).
\end{equation*}
We then get
\begin{equation*}
\gamma_{\bar R}=\gamma+\gamma_{\bar R}'-\frac{1}{2}\sum_{i=1}^N\sum_{j=1}^N E_{ij}^{\bar R}(|u_{i}\rangle\langle u_{j,\bar R}|+|u_{j,\bar R}\rangle\langle u_{i}|)+O(\bar R^{-8}),
\end{equation*}
where $\gamma_{\bar R}'=\sum_{i=1}^N k_i|u_{i,\bar R}\rangle \langle u_{i,\bar R}|.$
A simple computation gives
\begin{equation}\label{2.7}
\begin{aligned}
\text{Tr}(\sqrt{-\Delta}{\gamma_{\bar R}})=&\text{Tr}(\sqrt{-\Delta}{\gamma})+\text{Tr}(\sqrt{-\Delta}{\gamma_{\bar R}'})-\sum_{i=1}^N\sum_{j=1}^N E_{ij}^{\bar R}\Big(\mu_iE_{ij}^{\bar R}
+\frac{2}{\alpha}\int_{\mathbb{R}^3}(\rho_\gamma*|x|^{-\alpha})u_iu_{j,\bar R}dx\Big)
+O(\bar R^{-8})\\
\le& 2+\frac{2}{\alpha}\sum_{i=1}^N\sum_{j=1}^NE_{ij}^{\bar R}\|\rho_\gamma*|x|^{-\alpha}\|_{\infty}\int_{\mathbb{R}^3}u_iu_{j,\bar R}dx+O(\bar R^{-8})
\le 2+O(\bar R^{-8})
\end{aligned}
\end{equation}
and
\begin{equation}\label{2.8}
\rho_{\gamma_{\bar R}}=\rho_\gamma+\rho_{\gamma_{\bar R}'}-\sum_{i=1}^N\sum_{j=1}^NE_{ij}^{\bar R}u_i u_{j,\bar R}+O(\bar R^{-8}).
\end{equation}
It follows from (\ref{2.8}) that
\begin{equation}\label{2.9}
\begin{aligned}
\int_{\mathbb{R}^3}(\rho_{\gamma_{\bar R}}*|x|^{-\alpha})\rho_{\gamma_{\bar R}}=&2\int_{\mathbb{R}^3}(\rho_\gamma*|x|^{-\alpha})\rho_\gamma dx+2\sum_{i=1}^N\sum_{j=1}^N\int_{\mathbb{R}^3}(u_i^2*|x|^{-\alpha})u_{j,\bar R}^2dx\\
&-\sum_{i=1}^N\sum_{j=1}^N\int_{\mathbb{R}^3}(\rho_\gamma*|x|^{-\alpha})E_{ij}^{\bar R}u_i u_{j,\bar R}dx+O(\bar R^{-8})\\
\ge&2\int_{\mathbb{R}^3}(\rho_\gamma*|x|^{-\alpha})\rho_\gamma dx+2\int_{\mathbb{ R}^3}(u_1^2*|x|^{-\alpha})u_{1,\bar R}^2dx\\
&-C\|\rho_\gamma*|x|^{-\alpha}\|_{\infty}{\bar R}^{-8}+O({\bar R}^{-8}).
\end{aligned}
\end{equation}
Using Theorem \ref{tha} in Appendix, we can now derive estimate about $u_1$ that
\begin{equation*}
\int_{\mathbb{R}^3}\int_{\mathbb{R}^3}\frac{u_1^2(x)u_1^2(y-\bar Re_1)}{|x-y|^{\alpha}}dx \ge \int_{|x|\le \bar R}\int_{\bar R\le |y-\bar Re_1|\le 2\bar R}\frac{u_1^2(x)u_1^2(y-\bar Re_1)}{|x-y|^{\alpha}}dxdy.
\end{equation*}
It follows from triangle inequality $|x-y|\le |x|+|y|\le 4\bar R$ in the domain $\big\{(x,y)\big| |x|\le \bar R,\ \bar R\le|y-\bar Re_1|\le 2\bar R\big\}$ that
\begin{equation}\label{2.10}
\begin{aligned}
\int_{\mathbb{R}^3}\int_{\mathbb{R}^3}\frac{u_1^2(x)u_1^2(y-\bar Re_1)}{|x-y|^{\alpha}}dx \ge& \frac{1}{(4\bar R)^{\alpha}} \int_{|x|\le \bar R}\int_{\bar R\le |y-\bar Re_1|\le 2\bar R}u_1^2(x)u_1^2(y-\bar Re_1)dxdy\\
\ge & \frac{C}{\bar R^\alpha}\int_{|x|\le \bar R}u_1^2dx\int_{\bar R\le |y|\le 2\bar R}\frac{1}{|y|^8}dy\\
\ge & \frac{C}{\bar R^\alpha}\Big(\int_{\mathbb{R}^3}u_1^2dx-\int_{|x|\ge \bar R}u_1^2dx\Big)\frac{1}{\bar R^5}\\
\ge & \frac{C}{\bar R^{5+\alpha}}-\frac{C}{\bar R^{5+\alpha}}\int_{\bar R}^{\infty}\frac{1}{|x|^8}dx\ge\frac{C}{\bar R^{5+\alpha}}.
\end{aligned}
\end{equation}

Therefore, taking (\ref{2.6.1}), (\ref{2.7}), (\ref{2.9}) and (\ref{2.10}), we have
\begin{equation*}
\begin{aligned}
(\mathcal{K}_{\alpha,q}^{(2N)})^\alpha&\le (\mathcal{K}_{\alpha,q}^{(N)})^\alpha\frac{2^{\frac{2-\alpha}{q}}(2+O(\bar R^{-8}))^\alpha}{2+C\bar R^{-(5+\alpha)}}=2^{\frac{2-\alpha}{q}+\alpha-1}(\mathcal{K}_{\alpha,q}^{(N)})^{\alpha}(1-\frac{1}{2}C\bar R^{-(5+\alpha)}+O(\bar R^{-8}))< (\mathcal{K}_{\alpha,q}^{(N)})^{\alpha},
\end{aligned}
\end{equation*}
for $0<\alpha\le1$, $q$ is close to $\frac{2-\alpha}{(1-\alpha)_+}$ enough, and $\bar R$ is large enough. Hence, \eqref{eq2.12} is obtained.  The case $q=\frac{2-\alpha}{(1-\alpha)_+}$  can be proved  analogously, so we  omit the details here.

\textbf{(iv).} To finish the proof of (iv), it suffices to prove \eqref{eq1.8}. 
On the contrary, if $\mathcal{K}_{\alpha,q}^{(N)}=\mathcal{K}_{\alpha,q}^{(2N)}$ for some $N\in \mathbb{N}^+$. Then, there exists a $\gamma$ with ${\rm rank~}\gamma=M\leq N$ such that $\gamma$ is a minimizer of (\ref{1.1}) for $\mathcal{K}_{\alpha,q}^{(M)}=\mathcal{K}_{\alpha,q}^{(N)}=\mathcal{K}_{\alpha,q}^{(2N)}$. This indicates that $\mathcal{K}_{\alpha,q}^{(M)}=\mathcal{K}_{\alpha,q}^{(2M)}\ $, which however contradicts  \eqref{eq2.12}. 
\qed

\begin{remark}
According to the Gagliardo-Nirenberg type inequality \textnormal {\cite{fl}}, we have 
\begin{equation*}
\int_{\mathbb{R}^3}({|x|^{-\alpha}}*\rho )\rho dx\le C_{gn}\|(-\Delta)^{\frac{1}{4}}\sqrt{\rho}\|_2^{2\alpha}\|\rho\|_1^{2-\alpha}\leq C_{gn}N^\frac{(q-1)(2-\alpha)}{q}\|(-\Delta)^{\frac{1}{4}}\sqrt{\rho}\|_2^{2\alpha}\|\gamma\|_{\mathcal{S}_q}^{2-\alpha}.
\end{equation*}
This indicates 
\begin{equation*}
\mathcal{K}_{\alpha,q}^{(N)}\ge C_{gn}^{-\frac{1}{\alpha}}N^{-\frac{(q-1)(2-\alpha)}{q\alpha}}>0, \quad \text{for all}\ N,
\end{equation*}
which gives a lower bound of $\mathcal{K}_{\alpha,q}^{(N)}$. Specifically,  this bound is uniformly w.r.t. $N$ for $q=1$.
\end{remark}


Motivated by  the arguments in Appendix A in \cite{fgl}, in the rest of this section,  we will finish the proof of Theorem \ref{dual} and show the dual relation between non-local Gagliardo-Nirenberg-Sobolev inequality and Lieb-Thirring type inequality.

\noindent\textit{ Proof of Theorem \ref{dual}.}
ssume that $\sqrt{-\Delta}+V*|x|^{-\alpha}$ has at least $N$ negative eigenvalues (counting multiplicity). Let $u_1,u_2,...,u_N$ be orthogonal eigenfunctions corresponding to the $N$ negative eigenvalues. Define an operator $$\gamma=\sum_{j=1}^N n_j|u_j\rangle \langle u_j|,\ \text{ with }\ n_j=|\lambda_j(\sqrt{-\Delta}+V*|x|^{-1})|^{\frac{1}{q-1}}, \ j=1,2,\cdots,N. $$ Set $H_{\alpha}:=\int_{\mathbb{R}^3}(\rho_\gamma*|x|^{-\alpha})\rho_\gamma dx$, it then follows that
\begin{equation*}
\begin{aligned}
\sum_{j=1}^N &|\lambda_j(\sqrt{-\Delta}+V(x)*|x|^{-\alpha})|^{q'}=-\sum_{j=1}^N n_j\int_{\mathbb{R}^3}\Big( |\sqrt{-\Delta} u_j|^2+(V(x)*|x|^{-\alpha})|u_j|^2\Big)dx\\
\le&-\mathcal{K}_{\alpha,q}^{(N)}\|\gamma\|_{\mathcal{S}^{q}}^{-\frac{2-\alpha}{\alpha}}H_{\alpha}^{\frac{1}{\alpha}}
+\int_{\mathbb{R}^3}(\rho_\gamma*|x|^{-\alpha})V_-(x)dx\\
\le&-\mathcal{K}_{\alpha,q}^{(N)}\|\gamma\|_{\mathcal{S}^{q}}^{-\frac{2-\alpha}{\alpha}}H_{\alpha}^{\frac{1}{\alpha}}+H_\alpha^{\frac{1}{2}}\Big(\int_{\mathbb{R}^3}(V_-*|x|^{-\alpha})V_-(x)dx\Big)^{\frac{1}{2}}.
\end{aligned}
\end{equation*}

Maximizing over $H_\alpha$, we have
\begin{equation*}
\sum_{j=1}^N |\lambda_j(\sqrt{-\Delta}+V(x)*|x|^{-\alpha})|^{q'}\le \frac{2-\alpha}{2}\Big(\frac{\alpha}{2}\Big)^{\frac{\alpha}{2-\alpha}}(\mathcal{K}_{\alpha,q}^{(N)})^{-\frac{\alpha}{2-\alpha}}\|\gamma\|_{\mathcal{S}^{q}}\Big(\int_{\mathbb{R}^3}(V_-*|x|^{-\alpha})V_-(x)dx\Big)^{\frac{1}{2-\alpha}}.
\end{equation*}
It follows from $\|\gamma\|_{\mathcal{S}^{q}}^{q}=\sum_{j=1}^N |\lambda_j(\sqrt{-\Delta}+V(x)*|x|^{-\alpha})|^{q'}$ that
\begin{equation*}
\sum_{j=1}^N |\lambda_j(\sqrt{-\Delta}+V(x)*|x|^{-\alpha})|^{q'}\le \Big(\frac{2-\alpha}{2}\Big)^{q'}\Big(\frac{\alpha}{2}\Big)^{\frac{\alpha q'}{2-\alpha}} (\mathcal{K}_{\alpha,q}^{(N)})^{-\frac{\alpha q'}{2-\alpha}}\Big(\int_{\mathbb{R}^3}(V_-*|x|^{-\alpha})V_-(x)dx\Big)^\frac{q'}{2-\alpha}.
\end{equation*}
This indicates  that the  Hartree type Lieb-Thirring inequality \ref{re3}   is well defined, and the optimal constant $\mathcal{L}_{\alpha,q'}^{(N)}$ satisfies  
\begin{equation}\label{eq-dual}
\Big(\frac{2-\alpha}{2}\Big)^{\frac{q}{q-1}}\Big(\frac{\alpha}{2}\Big)^{\frac{\alpha q}{(2-\alpha)(q-1)}} (\mathcal{K}_{\alpha,q}^{(N)})^{-\frac{\alpha q}{(2-\alpha)(q-1)}}\ge \mathcal{L}_{\alpha,q'}^{(N)}.
\end{equation}

On the other hand, for any given  operator $0\leq \gamma=\sum_{i=1}^Nn_i|u_i\rangle \langle u_i|$,  
we choose $V(x)=-\beta \rho_\gamma$ with $\beta>0$ to be determined. Then,
\begin{equation*}
\begin{aligned}
&\sum_{j=1}^N n_j\int_{\mathbb{R}^3}|(-\Delta)^{\frac{1}{4}}u_j|^2dx-\beta\int_{\mathbb{R}^3}(\rho_\gamma*|x|^{-\alpha})\rho_\gamma dx
=\sum_{j=1}^N n_j \int_{\mathbb{R}^3}\Big(|(-\Delta)^{\frac{1}{4}}u_j|^2dx-\beta(\rho_\gamma*|x|^{-\alpha})|u_j|^2\Big) dx\\
&\ge -\|\gamma\|_{\mathcal{S}^{q}}\Big(\sum_{j=1}^N |\lambda_j(\sqrt{-\Delta}-\beta\rho_\gamma*|x|^{-\alpha})|^{q'}\Big)^{\frac{1}{q'}}
\ge -\|\gamma\|_{\mathcal{S}^{q}}(\mathcal{L}_{\alpha,q'}^{(N)})^{\frac{1}{q'}}\beta^{\frac{2}{2-\alpha}}\Big(\int_{\mathbb{R}^3}(\rho_\gamma*|x|^{-\alpha})\rho_\gamma dx\Big)^{\frac{1}{2-\alpha}},
\end{aligned}
\end{equation*}
where we used H\"older inequality in Schatten space \cite{s}  in the first ``$\ge$”. Still denote $H_{\alpha}=\int_{\mathbb{R}^3}(\rho_\gamma*|x|^{-\alpha})\rho_\gamma dx$, we then get that 
\begin{equation*}
\sum_{j=1}^N n_j\int_{\mathbb{R}^3}|(-\Delta)^{\frac{1}{4}}u_j|^2dx\ge \beta H_{\alpha}-\|\gamma\|_{\mathcal{S}^{q}}(\mathcal{L}_{\alpha,q'}^{(N)})^{\frac{1}{q'}}\beta^{\frac{2}{2-\alpha}}H_{\alpha}^{\frac{1}{2-\alpha}}.
\end{equation*}
Optimizing over $\beta$, we obtain
\begin{equation*}
\sum_{j=1}^N n_j\int_{\mathbb{R}^3}|(-\Delta)^{\frac{1}{4}}u_j|^2dx\ge \Big(\frac{2-\alpha}{2}\Big)^{\frac{2-\alpha}{\alpha}}\Big(\frac{\alpha}{2}\Big) (\mathcal{L}_{\alpha,q'}^{(N)})^{-\frac{2-\alpha}{\alpha q'}} \|\gamma\|_{\mathcal{S}^{q}}^{-\frac{2-\alpha}{\alpha}}H_{\alpha}^{\frac{1}{\alpha}}.
\end{equation*}
It then follows from the GNS inequality \eqref{1.1} that 
\begin{equation*}
\frac{\sum_{j=1}^N n_j\int_{\mathbb{R}^3}|(-\Delta)^{\frac{1}{4}}u_j|^2dx \|\gamma\|_{\mathcal{S}^{q}}^{\frac{2-\alpha}{\alpha}} }{H_{\alpha}^{\frac{1}{\alpha}}}\ge \mathcal{K}_{\alpha,q}^{(N)}\ge \Big(\frac{2-\alpha}{2}\Big)^{\frac{2-\alpha}{\alpha}}\Big(\frac{\alpha}{2}\Big) (\mathcal{L}_{\alpha,q'}^{(N)})^{-\frac{(2-\alpha)(q-1)}{\alpha q}}.
\end{equation*}
Comparing it with \eqref{eq-dual}, we obtain that
\begin{equation*}
\mathcal{K}_{\alpha, q}^{(N)}(\mathcal{L}_{\alpha,q'}^{(N)})^{\frac{(2-\alpha)(q-1)}{\alpha q}}=\frac{\alpha}{2}\Big(\frac{2-\alpha}{2}\Big)^{\frac{2-\alpha}{\alpha}}.
\end{equation*}
\qed

\section{Existence and nonexistence for problem (\ref{mini-N})} \label{Sec-3}
This section is devoted to the proof of Theorem \ref{th2}. We first give the existence and nonexistence for problem (\ref{mini-N}) by using Theorem \ref{th1} with $\alpha=1$ and $q=\infty$, denote $\mathcal{K}_\infty^{(N)}:=\mathcal{K}_{1,\infty}^{(N)}$, and then the asymptotic behavior of optimizer will be considered in the next section.

\noindent\textit{ Proof of Theorem \ref{th2}.} \text{\bf (i). Existence for $K\in (0,\mathcal{K}_\infty^{(N)})$.} For $K\in (0,\mathcal{K}_\infty^{(N)})$, we infer from Theorem \ref{th1} that
\begin{equation}\label{3.0}
\begin{aligned}
E_K(N)&\ge \text{Tr}[(\sqrt{-\Delta}+V(x))\gamma]-K\int_{\mathbb{R}^3}(\rho_\gamma*|x|^{-1})\rho_\gamma dx\ge \Big(1-\frac{K}{\mathcal{K}_\infty^{(N)}}\Big)\text{Tr}(\sqrt{-\Delta}\gamma)+\int_{\mathbb{R}^3}V(x)\rho_\gamma dx.
\end{aligned}
\end{equation}
So $E_K(N)$ is bounded from below for $K\in(0,\mathcal{K}_\infty^{(N)}]$. Let   $\{\gamma_n=\sum_{i=1}^{r_n} |u_i^n\rangle \langle u_i^n|\}_n$ be  a minimizing sequence of $E_K(N)$, where $1\leq r_n\leq N$ is an integer  and  $(u_i^n,u_j^n)_{L^2}=\delta_{ij}$ for $1\leq i,j\leq r_n$. Since $1\leq r_n\leq N$ is an integer, up to a subsequence, we may assume that $r_n\equiv r_K\in [1,N]$.   By Hoffmann-Ostenhof type inequality \cite{fl}
\begin{equation*}
\text{Tr}(\sqrt{-\Delta}\gamma_n)\ge \int_{\mathbb{R}^3}|(-\Delta)^{\frac{1}{4}}\sqrt{\rho_{\gamma_n}}|^2dx.
\end{equation*}
From (\ref{3.0}) we know that $\{u_i^n\}_n^\infty$ is uniformly bounded in $\mathcal{H}$ for all $i=1,2,\cdots, r_K$. Applying the compact embedding theorem in \cite{af}, we obtain
\begin{equation*}
\begin{aligned}
&u_i^n\overset{n}\rightharpoonup u_i\ \text{weakly in}\ \mathcal{H}, 
u_i^n\overset{n}\to u_i\ \text{strongly in}\ L^s(\mathbb{R}^3)\ \forall 2\le s<3, \text{\ and \ }(u_i,u_j)_{L^2}=\delta_{ij} \text{\  for \ }i,j=1,2,\cdots,r_K.
\end{aligned}
\end{equation*}
Denote $\gamma_K:=\sum_{i=1}^{r_K} |u_i\rangle \langle u_i|$, then we have 
\begin{equation*}
 \mathcal{E}_K(\gamma_K)\ge E_K(N)=\lim_{n\to\infty}\mathcal{E}_K(\gamma_n)\ge \mathcal{E}_K(\gamma_K).
\end{equation*}
This means $\gamma_K$ is a minimizer of (\ref{mini-N}).

Similar to the argument in Appendix A in \cite{s1}, $\gamma_K$ can be rewritten in the form $\gamma_K=\sum_{i=1}^{r_K}|u_{l_i}\rangle\langle u_{l_i}|$ where $u_{l_i}$ is $l_i$-th eigenfunction of the operator
\begin{equation*}
H_V u_{l_i}:=[\sqrt{-\Delta+m^2}+V(x)-2K\rho_{\gamma_K}*|x|^{-1}]u_{l_i}=\mu_{l_i}u_{l_i}
\end{equation*}
and $\rho_{\gamma_K}=\sum_{i=1}^{r_K}|u_{l_i}|^2$. We first prove that $\mu_{l_1}=\mu_1$, on the contrary, suppose $\mu_{l_1}\not=\mu_1$, we construct an operator as follows
\begin{equation*}
\gamma':=\gamma_K-|u_{l_1}\rangle \langle u_{l_1}|+|u_1\rangle\langle u_1|.
\end{equation*}
It follows that
\begin{equation}\label{3.1}
\begin{aligned}
\text{Tr}(\sqrt{-\Delta+m^2}\gamma')=&\T(\sqrt{-\Delta+m^2}\gamma_K)-(\sqrt{-\Delta+m^2}u_{l_1}, u_{l_1})+(\sqrt{-\Delta+m^2}u_1, u_1)\\
=&\text{Tr}(\sqrt{-\Delta+m^2}\gamma_K)+2K\int_{\mathbb{R}^3}(\rho_{\gamma_K}*|x|^{-1})(|u_1|^2-|u_{l_1}|^2)\\
&+\int_{\mathbb{R}^3}V(x)(|u_{l_1}|^2-|u_1|^2)dx+\mu_1-\mu_{l_1}
\end{aligned}
\end{equation}
and
\begin{equation}\label{3.2}
\text{Tr}(V(x)\gamma')=\text{Tr}(V(x)\gamma_K)+\int_{\mathbb{R}^3}V(x)(|u_1|^2-|u_{l_1}|^2)dx.
\end{equation}
Similarly,
\begin{equation}\label{3.3}
\begin{aligned}
\int_{\mathbb{R}^3}(\rho_{\gamma'}*|x|^{-1})\rho_{\gamma'}dx=&\int_{\mathbb{R}^3}(\rho_{\gamma_K}*|x|^{-1})\rho_{\gamma_K}dx+2\int_{\mathbb{R}^3}(\rho_{\gamma_K}*|x|^{-1})(|u_1|^2-|u_{l_1}|^2)dx\\
&+\int_{\mathbb{R}^3}\big[(|u_1|^2-|u_{l_1}|^2)*|x|^{-1}\big](|u_1|^2-|u_{l_1}|^2)dx\\
\ge&\int_{\mathbb{R}^3}(\rho_{\gamma_K}*|x|^{-1})\rho_{\gamma_K}dx+2\int_{\mathbb{R}^3}(\rho_{\gamma_K}*|x|^{-1})(|u_1|^2-|u_{l_1}|^2)dx.
\end{aligned}
\end{equation}
 Plugging (\ref{3.1}), (\ref{3.2}) and (\ref{3.3}) back into $\mathcal{E}_K(\gamma')$, we get the followoing contradiction:
\begin{equation*}
E_K(N)\le \mathcal{E}_K(\gamma')\le \mathcal{E}_K(\gamma_K)+\mu_1-\mu_{k_1}<\mathcal{E}_K(\gamma)=E_K(N).
\end{equation*}
Proceeding the above arguments, we can see that $\mu_{l_i}=\mu_i$, $\forall\ i=1,2,\cdots, r_K$. 

Finally, we claim that if $\mu_i<0$, $\forall\ 1\le i\le r_K$ and $r_K<N$, then $\mu_{r_K+1}>0$.
For otherwise, assume $u_{r_K+1}$ is the eigenfunction of $H_V$ corresponding to $\mu_{r_K+1}\le0$, set 
$\gamma':=\gamma_K+|u_{r_K+1}\rangle \langle u_{r_K+1}|$, then similar to \eqref{3.1}-\eqref{3.3}, one can deduce the following contradiction:
$$ E_K(N)\le \mathcal{E}_K(\gamma')=\mathcal{E}_K(\gamma_K)+\mu_{r_K+1}-K\int_{\mathbb{R}^3}(u_{r_K+1}^2*|x|^{-1})u_{r_K+1}^2dx<E_K(N).$$

\text{\bf(ii). Rank of minimizer} 
let $\mu_{V}>0$ be the first  eigenvalue of $\sqrt{-\Delta+m^2}+V(x)$,  and $v(x)$ be the corresponding eigenfunctions. 

We first show that as $K\searrow 0$, then $\Rank \gamma_K\equiv1$.  Indeed, by taking $\gamma_1:=|v\rangle\langle v|$ as a trial operator, one can deduce that 
\begin{equation*}
E_K(N)\le \mathcal{E}(\gamma_1)= \mu_V-K\int_{\mathbb{R}^3}(v^2*|x|^{-1})v^2dx\le \mu_V+O(K) \quad \text{ as\ } K\searrow0.
\end{equation*}
On the contrary, assume that the minimizer $\gamma_K$ for (\ref{mini-N})  satisfies $\Rank \gamma_K\ge 2$.  Let $\gamma_K$  be the minimizer for (\ref{mini-N}).  From \eqref{1.1}, one can see that  there exists $C>0$ independent of $K\searrow0$ such that 
\begin{equation*}
\int_{\mathbb{R}^3}(\rho_{\gamma_K}*|x|^{-1})\rho_{\gamma_K} dx\le C\quad \text{as } K\searrow 0. 
\end{equation*}
 If $\Rank \gamma_K\geq2$,  we then deduce that 
\begin{equation*}
E_K(N)=\T(\sqrt{-\Delta+m^2}\gamma_K)+\int_{\mathbb{R}^3}V(x)\rho_{\gamma_K} dx+O(K)\ge 2\mu_V+O(K) \quad \text{as\ } K\searrow0 ,
\end{equation*}
which obviouly derives a contradiction. Thus,  $\Rank\gamma_K\equiv1$ provided $K>0$ is small enough. Moreover, from Appendix A in \cite{gzz}, we know that the  minimizer for  (\ref{mini-N})  is unique when $K>0$ is small enough. 

We next focus on the case that $K\nearrow \mathcal{K}_\infty^{(N)}$. We intend to show that $E_K(N)\to 0$ as $K\nearrow \mathcal{K}_\infty^{(N)}$.  Let $\gamma_0=\sum_{i=1}^r |Q_i\rangle \langle Q_i|$ with rank $\gamma_0=r$ be an optimizer for (\ref{1.1}) with $\mathcal{K}_\infty^{(N)}$, where $Q_i$ satisfies
\begin{equation*}
\sqrt{-\Delta} Q_i-2\mathcal{K}_\infty^{(N)}(\rho_{\gamma_0}*|x|^{-1})Q_i=\mu_iQ_i,\quad \forall 1\le i\le r.
\end{equation*}

Then we use cut-off function to construct
\begin{equation*}
Q_i^{\bar R}(x)=A_i^{\bar R}\bar R^{\frac{3}{2}}\phi(x-x_0)Q_i[\bar R(x-x_0)],
\end{equation*}
where $x_0$ is some point to be determined, $\phi$ is a smooth non-negative cut-off radial function such that $\phi(x)\equiv 1$ for $|x|\le 1$ and $\phi(x)\equiv 0$ for $|x|\ge 2$, and $A_i^{\bar R}$ is chosen such that $\|Q_i^{\bar R}\|_2=1$. We then estimate each $A_i^{\bar R}$ and $E_{ij}:=(Q_i^{\bar R},Q_j^{\bar R})_{L^2}$. By Theorem \ref{tha} we see that
\begin{equation*}
|A_i^{\bar R}-1|=\Big|\int_{\mathbb{R}^3\backslash{B_{\bar R}(x_0)}}(A_i^{\bar R})^2(\phi^2({{\bar R}}^{-1}x)-1)Q^2(x)dx\Big|\le C{\bar R}^{-5}
\end{equation*}
and
\begin{equation*}
\begin{aligned}
|E_{ij}|=|(Q_i^{\bar R},Q_j^{\bar R})_{L^2}|&=\int_{\mathbb{R}^3}{\bar R}^3Q_i({\bar R}(x-x_0))Q_j({\bar R}(x-x_0))dx+O({\bar R}^{-5})=\delta_{ij}+O({\bar R}^{-5}).
\end{aligned}
\end{equation*}
Furthermore, we can establish a Gram matrix
\begin{equation*}
G_{\bar R}:=\left( \begin{matrix}
1 &(Q_1^{\bar R}, Q_2^{\bar R})_{L^2}& \cdots & (Q_1^{\bar R}, Q_r^{\bar R})_{L^2} \\
(Q_2^{\bar R},Q_1^{\bar R})_{L^2} & 1& \cdots & (Q_2^{\bar R}, Q_r^{\bar R})_{L^2} \\
\vdots & \vdots& \ddots  & \vdots\\
(Q_r^{\bar R}, Q_1^{\bar R})_{L^2} & (Q_r^{\bar R}, Q_2^{\bar R})_{L^2} & \ldots & 1
\end{matrix}\right)_{r\times r}.
\end{equation*}
Taking $R$ large enough such that the Gram matrix is invertible, we define
\begin{equation*}
(\tilde Q_1^{\bar R},\tilde Q_2^{\bar R},...,\tilde Q_r^{\bar R}):=(Q_1^{\bar R},Q_2^{\bar R},...,Q_r^{\bar R})G_{\bar R}^{-\frac{1}{2}}.
\end{equation*}
It then follows that
\begin{equation*}
(\tilde Q_i^{\bar R},\tilde Q_j^{\bar R})_{L^2}=\delta_{ij},\quad \forall\ 1\le i,j\le r.
\end{equation*}
Set $$\tilde \gamma=\sum_{i=1}^r |\tilde Q_i^{\bar R} \rangle \langle \tilde Q_i^{\bar R}|.$$
Next, we first estimate potential term. Expanding the Gram matrix, we have
\begin{equation*}
\begin{aligned}
(\tilde Q_1^{\bar R},\tilde Q_2^{\bar R},...,\tilde Q_r^{\bar R})=&(Q_1^{\bar R},Q_2^{\bar R},...,Q_r^{\bar R}) +O(e_{\bar R}^2)\\
&-\frac{1}{2}\big(\sum_{i=2}^rE_{i1}Q_i^{\bar R}, \sum_{i=1,i\not =2}^r E_{i2}Q_i^{\bar R},\ldots,\sum_{i=1,i\not =j}^r E_{ij}Q_i^{\bar R},\ldots,\sum_{i=1}^{r-1}E_{ir}Q_i^{\bar R}\big),
\end{aligned}
\end{equation*}
where
$e_{\bar R}=\max_{i\not=j}|E_{ij}|=O({\bar R}^{-5})$.

We can derive for all $i=1,2,\cdots, r$,
\begin{equation*}
\begin{aligned}
\int_{\mathbb{R}^3}V(x)|\tilde Q_i^{\bar R}|^2dx&\le \int_{\mathbb{R}^3}V(x)[Q_i^{\bar R}-\sum_{i\not=j}\frac{1}{2}E_{ji} Q_j^{\bar R}+O(e_{\bar R}^2)]^2dx\\
&\le \int_{\mathbb{R}^3}V(\frac{x}{{\bar R}}+x_0)\phi^2(\frac{x}{{\bar R}})Q_i^2dx+O({\bar R}^{-5})\\
&\le V(x_0)+O({\bar R}^{-5})
\end{aligned}
\end{equation*}
and
\begin{equation}\label{3.3.1}
\T(\tilde \gamma V(x))\le V(x_0)\T(\gamma_0)+O({\bar R}^{-5}).
\end{equation}
Next, we will estimate the term with fractional Laplacian by following the ideas in \cite{yy},
\begin{equation*}
\int_{\mathbb{R}^3}|(-\Delta+m^2)^{\frac{1}{4}}\tilde Q_i^{\bar R}|^2dx={\bar R}\int_{\mathbb{R}^3}|(-\Delta+m^2{\bar R}^{-2})^{\frac{1}{4}}[\phi({\bar R}^{-1}x)Q_i]|^2dx+O({\bar R}^{-5}).
\end{equation*}
By Lemma 3 in \cite{ly1}, there holds
\begin{equation*}
\sqrt{-\Delta+{\bar R}^{-2}m^2}\le \sqrt{-\Delta} +\frac{1}{2}{\bar R}^{-2}m^2(-\Delta)^{-\frac{1}{2}}.
\end{equation*}
We now have
\begin{equation*}
\begin{aligned}
{\bar R}\int_{\mathbb{R}^3}|(-\Delta+m^2{\bar R}^{-2})^{\frac{1}{4}}\phi({\bar R}^{-1}x)Q_i|^2dx\le& {\bar R}\int_{\mathbb{R}^3}\phi({\bar R}^{-1}x)Q_i(x)(-\Delta)^{\frac{1}{2}}\big[\phi({\bar R}^{-1}x)Q_i(x)\big]dx\\
&+\frac{m^2}{2{\bar R}}\int_{\mathbb{R}^3}\phi({\bar R}^{-1}x)Q_i(x)(-\Delta)^{-\frac{1}{2}}[\phi({\bar R}^{-1}x)Q_i(x)]dx\\
:=&I+II.
\end{aligned}
\end{equation*}
We can divide I into three parts
\begin{equation*}
\begin{aligned}
I\le& {\bar R}\Big(\int_{\mathbb{R}^3}Q_i(x)\sqrt{-\Delta}Q_i(x)dx+\int_{\mathbb{R}^3}(\phi({\bar R}^{-1}x)-1)Q_i(x)\sqrt{-\Delta}Q_i(x)dx\\
&+\int_{\mathbb{R}^3}\phi({\bar R}^{-1}x)Q_i(x)\sqrt{-\Delta}\big[(\phi({\bar R}^{-1}x)-1)Q_i(x)\big]dx    \Big)\\
:=&{\bar R}\int_{\mathbb{R}^3}Q_i(x)\sqrt{-\Delta}Q_i(x)dx+I.I+I.II.
\end{aligned}
\end{equation*}
Due to $Q_i$ satisfied
\begin{equation*}
\sqrt{-\Delta} Q_i-2\rho_\gamma*|x|^{-1}Q_i=\mu_i Q_i,
\end{equation*}
it follows from Theorem \ref{tha}, we have
\begin{equation*}
|\sqrt{-\Delta} Q_i|\le (\frac{2C}{|x|}+|\mu_i|)|Q_i|.
\end{equation*}
It is easy to show that
\begin{equation*}
\begin{aligned}
|I.I|\le {\bar R}C\int_{\mathbb{R}^3\backslash B_{\bar R}(0)}|Q_i|^2dx
\le C{\bar R}^{-4}.
\end{aligned}
\end{equation*}
For term I.II, we need to use an estimate of commutators in \cite{ll1} that
\begin{equation*}
\big\|\big[\sqrt{-\Delta},\ \phi({\bar R}^{-1}x)\big]\big\|_{L^2 \to L^2}\le C\big\|\nabla \phi({\bar R}^{-1}x)\big\|_\infty.
\end{equation*}
Hence, we get
\begin{equation*}
\begin{aligned}
|I.II|\le&{\bar R}\Big|\int_{\mathbb{R}^3}(\phi({\bar R}^{-1}x)-1)Q(x)\Big(\phi({\bar R}^{-1}x)\sqrt{-\Delta}+\Big[\sqrt{-\Delta},\phi({\bar R}^{-1}x)\Big]\Big)Q_i(x)dx\Big|\\
\le& C{\bar R}\int_{\mathbb{R}^3\backslash B_{\bar R}(0)}|Q_i(x)||\sqrt{-\Delta}Q_i(x)|dx +\bar R\Big(\int_{\mathbb{R}^3\backslash B_{\bar R}(0)}Q^2(x)dx\Big)^{\frac{1}{2}}\big|\big|\big[\sqrt{-\Delta}, \phi({\bar R}^{-1}x)\big]\big|\big|_{L^2\to L^2}\|Q_i\|_2\\
\le& C{\bar R}^{-4}+C{\bar R}^{-\frac{5}{2}}=C{\bar R}^{-\frac{5}{2}}.
\end{aligned}
\end{equation*}
For the last term II, using Fourier transform and Plancherel Theorem, we have
\begin{equation*}
\begin{aligned}
|II|&=\frac{m^2}{2{\bar R}}\Big|\int_{\mathbb{R}^3}\big|(\phi(\frac{x}{{\bar R}})Q_i)^\wedge (\xi)\big|^2|\xi|^{-1}d\xi \Big|=\frac{C}{{\bar R}}\Big|\int_{\mathbb{R}^3}\Big[(\phi(\frac{x}{{\bar R}})Q_i(x))*|x|^{-2}\Big]\phi(\frac{x}{{\bar R}})Q_i(x)dx\Big|.
\end{aligned}
\end{equation*}
It follows from Hardy-Littlewood-Sobolev inequality and Theorem \ref{tha} in the Appendix that
\begin{equation*}
|II|\le C{\bar R}^{-1}\|\phi(\frac{x}{{\bar R}})Q_i(x)\|_{\frac{3}{2}}^2\le C{\bar R}^{-1}.
\end{equation*}
In conclusion, we have
\begin{equation*}
\int_{\mathbb{R}^3}|(-\Delta+m^2)^{\frac{1}{4}}\tilde Q_i^{\bar R}|^2dx={\bar R}\int_{\mathbb{R}^3}|(-\Delta)^{\frac{1}{4}}Q_i|^2dx+O({\bar R}^{-1})
\end{equation*}
and
\begin{equation}\label{3.3.3}
\T(\sqrt{-\Delta+m^2}\tilde\gamma)=\bar R\T(\sqrt{-\Delta}\gamma_0)+O({\bar R}^{-1}).
\end{equation}

Another key step in the proof is to estimate the interaction terms. Expanding the interaction term, we have
\begin{equation*}
\begin{aligned}
&\int_{\mathbb{R}^3}\int_{\mathbb{R}^3}\frac{\rho_{\tilde \gamma}(x)\rho_{\tilde \gamma}(y)}{|x-y|}dxdy=\int_{\mathbb{R}^3}\int_{\mathbb{R}^3}\frac{\sum_{i=1}^r(Q_i^{\bar R})^2(x)\sum_{j=1}^r(Q_j^{\bar R})^2(y)}{|x-y|}dxdy+O({\bar R}^{-5})\\
=&\sum_{i=1}^r\int_{\mathbb{R}^3}\int_{\mathbb{R}^3}\frac{(Q_i^{\bar R})^2(x)(Q_i^{\bar R})^2(y)}{|x-y|}dxdy+2\sum_{i=1}^r\sum_{j=1}^r\int_{\mathbb{R}^3}\int_{\mathbb{R}^3}\frac{(Q_i^{\bar R})^2(x)(Q_j^{\bar R})^2(y)}{|x-y|}dxdy
+O({\bar R}^{-5}).
\end{aligned}
\end{equation*}

For any $1 \le i,j\le N$, we have
\begin{equation}\label{3.3.5}
\begin{aligned}
\int_{\mathbb{R}^3}\int_{\mathbb{R}^3}\frac{(Q_i^{\bar R})^2(x)(Q_j^{\bar R})^2(y)}{|x-y|}dxdy=&{\bar R}\int_{\mathbb{R}^3}\int_{\mathbb{R}^3}\frac{\phi^2(\frac{x}{{\bar R}})Q_i^2(x)\phi^2(\frac{y}{{\bar R}})Q_i^2(y)}{|x-y|}dxdy+O({\bar R}^{-5})\\
=&{\bar R}\int_{\mathbb{R}^3}(Q_i^2*|x|^{-1})Q_j^2dx+{\bar R}\int_{\mathbb{R}^3}\frac{\phi^2(\frac{x}{{\bar R}})-1}{|x-y|}Q_i^2(x)Q_j^2(y)dxdy\\
&+{\bar R}\int_{\mathbb{R}^3}\int_{\mathbb{R}^3}\frac{(\phi^2(\frac{y}{{\bar R}})-1)\phi^2(\frac{x}{{\bar R}})}{|x-y|}Q_i^2(x)Q_j^2(y)dxdy+O({\bar R}^{-5})\\
=&{\bar R}\int_{\mathbb{R}^3}(Q_i^2*|x|^{-1})Q_j^2dx+III+IV+O({\bar R}^{-5}).
\end{aligned}
\end{equation}
It follows from Theorem \ref{tha}, Newton theorem and Hardy-Kato inequality and that
\begin{equation}\label{3.4}
|III|\le {\bar R}\int_{\mathbb{R}^3}\frac{Q_j^2(y)}{|y|}dy\int_{\mathbb{R}^3}\frac{C|\phi^2({\bar R}^{-1}x)-1|}{1+|x|^8}dx
\le C{\bar R}\int_{\mathbb{R}^3\backslash B_{\bar R}(0)}\frac{C}{1+|x|^8}dx\le C{\bar R}^{-4}
\end{equation}
and
\begin{equation}\label{3.5}
|IV|\le{\bar R}\int_{\mathbb{R}^3}\frac{\phi^2({\bar R}^{-1}x)Q_i^2(x)}{|x|}dx\int_{\mathbb{R}^3}\frac{C|\phi^2({\bar R}^{-1}y)-1|}{|1+|y|^8}dy
\le C{\bar R}^{-4}.
\end{equation}
Taking (\ref{3.4}) and (\ref{3.5}) into (\ref{3.3.5}), we obtain that
\begin{equation*}
\int_{\mathbb{R}^3}\int_{\mathbb{R}^3}\frac{(Q_i^{\bar R})^2(x)(Q_j^{\bar R})^2(y)}{|x-y|}dxdy={\bar R}\int_{\mathbb{R}^3}(Q_i^2*|x|^{-1})Q_j^2dx+O({\bar R}^{-4})
\end{equation*}
and
\begin{equation}\label{3.6}
\int_{\mathbb{R}^3}\int_{\mathbb{R}^3}\frac{\rho_{\tilde \gamma}(x)\rho_{\tilde \gamma}(y)}{|x-y|}dxdy=\bar R\int_{\mathbb{R}^3}\int_{\mathbb{R}^3}\frac{\rho_{ \gamma_0}(x)\rho_{\gamma_0}(y)}{|x-y|}dxdy+O({\bar R}^{-4}).
\end{equation}

To summarize, substituting (\ref{3.3.1}), (\ref{3.3.3}) and (\ref{3.6}) into $\mathcal{E}_K(\tilde \gamma)$, we get
\begin{equation}\label{eq3.07}
\begin{aligned}
\mathcal{E}_K(\tilde \gamma)&\le {\bar R}\Big[\text{Tr}(\sqrt{-\Delta}\gamma_0)-K\int_{\mathbb{R}^3}\rho_{\gamma_0}*|x|^{-1}\rho_{\gamma_0}dx\Big]+V(x_0)\text{Tr}(\gamma_0) +O({\bar R}^{-1})\\
&\le {\bar R}\frac{\mathcal{K}_\infty^{(N)}-K}{\mathcal{K}_\infty^{(N)}}\T(\sqrt{-\Delta}\gamma_0)+rV(x_0)+C\mathcal{K}_\infty^{(N)}\bar R^{-1}.
\end{aligned}
\end{equation}
 Choosing $x_0\in \mathbb{R}^3$ such that $V(x_0)=0$ and $$\bar R=\mathcal{K}_\infty^{(N)}\sqrt{\frac{C\T(\sqrt{-\Delta}\gamma_0)}{\mathcal{K}_\infty^{(N)}-K}}\to\infty \quad\text{ as $K\nearrow \mathcal{K}_\infty^{(N)}$},$$  then we deduce that 
\begin{equation}\label{3.7}
\begin{split}
    E_K(N)&\leq \mathcal{E}_K(\tilde \gamma) \le {\bar R}\frac{\mathcal{K}_\infty^{(N)}-K}{\mathcal{K}_\infty^{(N)}}\T(\sqrt{-\Delta}\gamma_0)+C\mathcal{K}_\infty^{(N)}\bar R^{-1}= 2\sqrt{C\T(\sqrt{-\Delta}\gamma_0})\sqrt{\mathcal{K}_\infty^{(N)}-K}\\
    &=O\big((\mathcal{K}_\infty^{(N)}-K)^\frac12\big) \to 0\quad \text{as\ } K\nearrow \mathcal{K}_\infty^{(N)}. 
\end{split}
\end{equation}

If $\Rank \gamma_K\le [\frac{N}{2}]$, from (iv) in Theorem \ref{th1} we know that $ \mathcal{K}_\infty^{([\frac{N}{2}])}>\mathcal{K}_\infty^{(N)}$.   then we have
\begin{equation*}
\begin{split}
    \mathcal{E}_K(N)&\ge (1-\frac{K}{\mathcal{K}_\infty^{([\frac{N}{2}])}})\T(\sqrt{-\Delta+m^2}\gamma_K)+\int_{\mathbb{R}^3}V(x)\rho_{\gamma_K} dx\ge \frac{\mathcal{K}_\infty^{([\frac{N}{2}])}-\mathcal{K}_\infty^{(N)}}{\mathcal{K}_\infty^{([\frac{N}{2}])}}\Big(\T(\sqrt{-\Delta+m^2}\gamma_K)+\int_{\mathbb{R}^3}V(x)\rho_{\gamma_K} dx\Big)\\
    &\geq \frac{\mathcal{K}_\infty^{([\frac{N}{2}])}-\mathcal{K}_\infty^{(N)}}{\mathcal{K}_\infty^{([\frac{N}{2}])}}\mu_V.
\end{split}
\end{equation*}
This contradicts \eqref{3.7}.

\text{\bf (iii). Nonexistence for $K\in [\mathcal{K}_\infty^{(N)},+\infty)$.}
When  $K> \mathcal{K}_\infty^{(N)}$ is fixed, we deduce that  $E_K(N)=-\infty $ by taking $\bar R\to\infty$ in (\ref{eq3.07}), which means (\ref{mini-N}) has no minimizer. 

When $K=\mathcal{K}_\infty^{(N)}$, we see from \eqref{3.7} that  $E_{{K}_\infty^{(N)}}(N)=0$. On the contrary, if there is a minimizer $\gamma_0$ of $E_{\mathcal{K}_\infty^{(N)}}$, we deduce that 
\begin{equation*}
\int_{\mathbb{R}^3}V(x)\rho_{\gamma_0}(x)dx=0
\end{equation*}
and 
\begin{equation*}
\text{Tr}(\sqrt{-\Delta}\gamma_0)=\mathcal{K}_\infty^{(N)}\int_{\mathbb{R}^3}(\rho_{\gamma_0}*|x|^{-1})\rho_{\gamma_0}dx.
\end{equation*}
This leads to a contradiction, since the first identity indicates that 
$\rho_{\gamma_0}$ must be a compact support function for $V(x)\to \infty$ as $|x|\to \infty$. However, the second identity indicates that $\gamma_0$ is indeed an optimizer for GNS-inequality \eqref{1.1}, and thus   $\rho_{\gamma_0}(x)>0$ in $\mathbb{R}^3$. \qed

\section{Asymptotic behavior of minimizers} \label{Sec-4}

In this section, we investigate the asymptotic behavior of minimizers for (\ref{mini-N}) as $K\nearrow \mathcal{K}_\infty^{(N)}$. We first establish rough blow-up properties for general potentials, then derive quantitative estimates of minimizers under additional assumptions on the potential.

\begin{lemma}\label{lem3.1}
Let $\gamma_K$ be a minimizer for (\ref{mini-N}), then $\T (\sqrt{-\Delta}\gamma_K)\to \infty$, as $K\nearrow \mathcal{K}_\infty^{(N)}$.
\end{lemma}
\begin{proof}

We argue by contradiction. Assume that there exists a subsequence $K_k\nearrow \mathcal{K}_\infty^{(N)}$ as $k\to\infty$, such that    $$\gamma_{K_k}:=\gamma_{k}=\sum_{i=1}^{r_k} |u_i^{k}\rangle\langle u_i^{k}|$$
satisfies  $\T(\sqrt{-\Delta}\gamma_k)\le C$. Since $1\leq r_k\leq N$ is an integer, up to a subsequence, we may assume that $r_k\equiv r\in [1,N]$. Then, $\int_{\mathbb{R}^3}V(x)\rho_{\gamma_k}dx\le C$ follows by (\ref{3.0}).
Thus, for every $1\le i\le r$, there exists a bounded sequence $\{u_i^{k}\}$ in $\mathcal{H}$ such that
\begin{equation*}
u_i^{k} \overset{k}\to u_i^0,\ \text{in}\ L^q(\mathbb{R}^3),\ \text{for}\ 2\le q<3.
\end{equation*}
Define $\tilde \gamma_{0}=\sum_{i=1}^r|u_i^0\rangle \langle u_i^0|$, we see that 
\begin{equation*}
\int_{\mathbb{R}^3} (\rho_{\gamma_{k}}*|x|^{-1}) \rho_{\gamma_{k}}dx\to \int_{\mathbb{R}^3} (\rho_{\tilde \gamma_{0}}*|x|^{-1}) \rho_{\tilde \gamma_{0}} dx,\quad \text{as}\ k\to\infty
\end{equation*}
and then 
\begin{equation*}
0=E_{\mathcal{K}_\infty^{(N)}}(N)\le \mathcal{E}_{\mathcal{K}_\infty^{(N)}}(\tilde \gamma_{0})=\lim_{k\to\infty}\mathcal{E}_{K_k}(\gamma_{k})=0.
\end{equation*}
This indicates that $\tilde \gamma_{0}$ is a minimizer of $E_{\mathcal{K}_\infty^{(N)}}$,  which however contradicts \textbf{(ii)} of Theorem \ref{th2}.
\end{proof}

After the preparations, we will show the asymptotic behavior of minimizer.\\ 

\noindent\textit{ Proof of Theorem \ref{th3}.}
It follows from Lemma \ref{lem3.1} that there exist $K_k\nearrow \mathcal{K}_\infty^{(N)}$ as $k\to\infty$ such that
\begin{equation*}
\text{Tr}(\sqrt{-\Delta}\gamma_{k})\to \infty,\ \text{as}\ k\to \infty,
\end{equation*}
where $\gamma_{k}=\sum_{i=1}^{r_k} |u_i^k \rangle\langle u_i^k|$. Since $1\leq r_k\leq N$ is an integer, up to a subsequence, we may assume that $r_k\equiv r\in [1,N]$. We set $\rho_k:=\rho_{\gamma_k}$ and deduce from 
\eqref{1.1} that 
\begin{equation}\label{eq-4.1}
\begin{aligned}
\lim_{k\to\infty}\big[\text{Tr}(\sqrt{-\Delta} \gamma_{k})-K_k\int_{\mathbb{R}^3} (\rho_{k}*|x|^{-1})\rho_{k}dx\big]=0 \ \text{ and }\lim_{k\to\infty}\int_{\mathbb{R}^3}V(x) \rho_{k}dx=0.
\end{aligned}
\end{equation}
Let
\begin{equation*}
\tilde \gamma_k=\sum_{i=1}^r |w_i^k\rangle \langle w_i^k| \ \text{ with }\ w_i^k=\varepsilon_k^\frac{3}{2} u_i^k(\varepsilon_k x+z_k),
\end{equation*}
where $\varepsilon_k\overset{k}{\to}0^+$ is given by \eqref{eq-eps} and    $z_k\in \mathbb{R}^3$  is to be determined later.  It then follows  that  
\begin{equation}\label{4.1}
\text{Tr}(\sqrt{-\Delta}\tilde \gamma_k)\equiv1,\quad K_k\int_{\mathbb{R}^3}(\tilde \rho_{k}*|x|^{-1})\tilde \rho_{k}dx\to 1,\ \text{as}\ k\to \infty,
\end{equation}
where $\tilde \rho_k:=\rho_{\tilde \gamma_k}$.
From Hoffmann-Ostenhof type inequality we know that $\sqrt{\tilde \rho_{k}}$ is bounded in $H^{\frac{1}{2}}(\mathbb{R}^3)$. We claim that $\sqrt{\tilde \rho_{k}}$ is non-vanishing. In fact, if the assertion does not hold, then by vanishing lemma, we can derive that 
$\tilde \rho_{k}\overset{k}\to 0\ \text{in}\ L^q(\mathbb{R}^3),\forall\ 1< q<\frac{3}{2}.$
This further indicates 
$\int_{\mathbb{R}^3}(\tilde \rho_{k}*|x|^{-1})\tilde \rho_{k}dx\overset{k}\to 0,$
which contradicts (\ref{4.1}).
Thus, there exist $\{z_k\}$, and $M>0$ and $\beta\in(0,1)$ such that
\begin{equation}\label{eq-4.3}
\liminf_{K\nearrow \mathcal{K}_\infty^{(N)}}\int_{B_M(z_k)}\tilde \rho_{k}dx\ge \beta>0.
\end{equation}
Next, we intend to prove that  
\begin{equation*}
    \lim_{k\to\infty}\text{dist}(z_k,\Lambda )=0, \ \text{ where $\Lambda$ defined in (\ref{eq-1.15})},
\end{equation*} 
which indicates that, up to a subsequence, $z_k\overset{k}\to z_0\in \Lambda$.
  If not, there exists $\delta>0$ such that 
\begin{equation*}
\text{dist}(z_k,\Lambda )\ge \delta>0,\ \text{as}\ k\to \infty.
\end{equation*}
This together with $V\in C(\mathbb{R}^3)$ implies there exists $C(\delta)>0$ such  that $V(z_k)\ge C(\delta)>0$ and then 
\begin{equation*}
\begin{aligned}
\liminf_{k\to\infty}\int_{\mathbb{R}^3}V(x) \rho_{k}dx=& \liminf_{k\to\infty}\int_{\mathbb{R}^3}V(\varepsilon_k x+z_k) \tilde \rho_{k}dx
\ge\int_{B_M(0)}\liminf_{k\to\infty}V(\varepsilon_k x+z_k)\tilde \rho_{k}dx
\ge \frac{C(\delta)\beta}{2}>0.
\end{aligned}
\end{equation*}
This  contradicts \eqref{eq-4.1}.   From \eqref{eq-4.3} we see  that there exists a positive operator $0\not =\gamma\in \mathcal{S}^1$ such that 
 \begin{equation*}
 \tilde \gamma_k\overset{\star}\rightharpoonup \gamma \ \text{ in } \mathcal{S}^1\ \text{ and } \rho_{\tilde \gamma_k} \to \rho_\gamma \text{ in } L_{\rm loc}^q (\mathbb{R}^3)\  \forall \ q\in[1,\frac32],\  \text{ as } k\to\infty.
\end{equation*}
As a consequence, there hold 
\begin{equation*}
\begin{aligned}
\T(\sqrt{-\Delta+m^2\varepsilon_k^2} \tilde \gamma_k)\geq \T(\sqrt{-\Delta} \tilde \gamma_k) 
=\T(\sqrt{-\Delta}\gamma)+\T\big(\sqrt{-\Delta} (\tilde \gamma_k-\gamma)\big)
\end{aligned}
\end{equation*}
and 
\begin{equation}\label{eq4.5}
\int_{\mathbb{R}^3}\Big(\tilde \rho_{k}*|x|^{-1}\Big)\tilde \rho_{k}dx=\int_{\mathbb{R}^3}(\rho_\gamma)*|x|^{-1})\rho_{\gamma}dx+\int_{\mathbb{R}^3}(\rho_{\tilde \gamma_k-\gamma})*|x|^{-1})\rho_{\tilde \gamma_k-\gamma}dx+o_k(1).
\end{equation}
Plugging the above two estimates into \eqref{eq-4.1}, and applying \eqref{1.1},  we get  that
\begin{equation}\label{eq4.05}
\begin{aligned}
0= &\lim_{k\to \infty} \Big[\text{Tr}(\sqrt{-\Delta}\tilde \gamma_k)-K_k\int_{\mathbb{R}^3}(\tilde \rho_{k}*{|x|^{-1}})\tilde \rho_{k}dx \Big]= \T(\sqrt{-\Delta}\gamma)-\mathcal{K}_\infty^{(N)}\int_{\mathbb{R}^3}(\rho_{ \gamma}*{|x|^{-1}})\rho_{ \gamma}dx\\&
+\lim_{k\to\infty}\Big[\text{Tr}\big(\sqrt{-\Delta}(\tilde \gamma_k-\gamma)\big)-K_k\int_{\mathbb{R}^3}( \rho_{\tilde\gamma_k-\gamma})*|x|^{-1}\rho_{\tilde\gamma_k-\gamma}dx\Big]\\
&\geq \T(\sqrt{-\Delta}\gamma)-\mathcal{K}_\infty^{(N)}\int_{\mathbb{R}^3}(\rho_{ \gamma}*{|x|^{-1}})\rho_{ \gamma}dx\\
&\geq \|\gamma\|\T(\sqrt{-\Delta}\gamma)-\mathcal{K}_\infty^{(N)}\int_{\mathbb{R}^3}(\rho_{ \gamma}*{|x|^{-1}})\rho_{ \gamma}dx\geq0,
\end{aligned}
\end{equation}
due to $\|\gamma\|\le \lim_{k\to \infty}\|\tilde \gamma_k\|=1$.
This together with the fact that $\gamma$ is finite rank operator implies that $\gamma\in\mathcal{S}^1$ with $\|\gamma\|=1$     is an optimizer of (\ref{1.1}) and 
\begin{equation}\label{eq-4.7}
    \lim_{k\to\infty}\Big[\text{Tr}\big(\sqrt{-\Delta}(\tilde \gamma_k-\gamma)\big)-K_k\int_{\mathbb{R}^3}( \rho_{\tilde\gamma_k-\gamma})*|x|^{-1}\rho_{\tilde\gamma_k-\gamma}dx\Big]=0.
\end{equation}
We denote   
$$\gamma=\sum_{i=1}^R|Q_i \rangle \langle Q_i|,\quad (Q_i,Q_j)=\delta_{ij},\ i,j=1,2,\cdots,R,  $$ where $[\frac{N}{2}]+1\le R\le r$  by \eqref{eq1.8}, and $Q_i$ is $i$-th eigenfunction of $H_{1,\gamma}$.

We claim that 
\begin{equation}\label{eq4.8}
    \T\big(\sqrt{-\Delta}(\tilde \gamma_k-\gamma)\big)=\T(\sqrt{-\Delta}\tilde \gamma_k)-\T(\sqrt{-\Delta}\gamma)=o_k(1).
\end{equation}
For otherwise, it follows  from \eqref{eq-4.7} that there exists $C>0$ such that 
$$\T\big(\sqrt{-\Delta}(\tilde \gamma_k-\gamma)\big), \int_{\mathbb{R}^3}( \rho_{\tilde\gamma_k-\gamma})*|x|^{-1}\rho_{\tilde\gamma_k-\gamma}dx\geq C>0. $$
Proceeding the same arguments between \eqref{eq-4.3} and \eqref{eq4.05},  one can derive that there exist $\{z'_k\}\subset\mathbb{R}^3$ satisfying $|z'_k|\overset{k}\to\infty$ and  $\tilde \gamma\in \mathcal{S}^1$ such that, by passing to a subsequence, 
\begin{equation*}
( \tilde \gamma_k-\gamma) (\cdot-z_k')\overset{\star}\rightharpoonup \tilde \gamma\not=0,\  \text{in}\ \mathcal{S}^{1}, \text{   as}\ k\to \infty
\end{equation*}
and $\tilde \gamma\in \mathcal{S}^1$ with  $\|\tilde \gamma\|=1$  is an optimizer of (\ref{1.1}).
Then, we can use \eqref{eq1.8} to deduce the follow contradiction:
\begin{equation*}
r\equiv \T(\tilde \gamma_k)=\int_{\mathbb{R}^3}\tilde \rho_kdx \ge \int_{\mathbb{R}^3}\rho_\gamma dx+ \int_{\mathbb{R}^3} \rho_{\tilde\gamma}dx=\T(\gamma)+\T(\tilde \gamma)> N\ge r.
\end{equation*}
From \eqref{eq4.5}, \eqref{eq-4.7} and \eqref{eq4.8}, we deduce that  \eqref{eq-zero} holds.
Finally, if $r=R$, it follows from  $\T \tilde \gamma_k=\T \gamma=r$ that  $\tilde \gamma_k\overset{k}\to \gamma $ in $\mathcal{S}^{1}$, and thus
\begin{equation*}
w_{i}^k(x)\to Q_i(x)\ \text{in}\ L^2(\mathbb{R}^3) \ \text{ for all }i=1,2,\cdots,r, \  \text{   as}\ \ k\to \infty. 
\end{equation*}
Then, one can further use \eqref{eq-4.7} to obtain \eqref{eq-zero2}.

\qed

\vskip .1truein
If we make further assumption (\ref{1.5}) for potential $V$,  we can prove that  $r=R$ holds for any $N\in\mathbb{N}^+$. Moreover,  the   energy $E_{K}(N)$ and the  blow-up rate of minimizers for \eqref{mini-N}  as $K_k\nearrow  \mathcal{K}_\infty^{(N)}$ can be calculated precisely.

\noindent\textit{Proof of Theorem \ref{th4}.}
We first establish a refined upper bound for $E_{K_k}(N)$.
Choose $x_j\in\mathcal{Z}$ and  $y\in\bar\Gamma$ with the sets $\mathcal{Z}$ and $\bar\Gamma$ being given by \eqref{eq-1.17} and \eqref{eq-1.18}, respectively. Recall  $\gamma\in \mathcal{S}^1$ is the optimizer of \eqref{1.1} obtained in Theorem \ref{th3}  and  set 
$$\gamma_{\tau_k}:=\sum_{i=1}^R |\tau_k^\frac{3}{2}Q_{i}(\tau_k (x-x_j)-y)\rangle\langle \tau_k^\frac{3}{2}Q_{i}(\tau_k (x-x_j)-y)|,$$
where  $\tau_k\overset{k}{\to}\infty$ will be determined later.  Direct calculations give that, as $\tau_k\to\infty$,
\begin{equation*}
\int_{\mathbb{R}^3}(\rho_{\gamma_{\tau_k}}*|x|^{-1})\rho_{ \gamma_{\tau_k}}dx=\tau_k\int_{\mathbb{R}^3}(\rho_{\gamma}*|x|^{-1})\rho_{ \gamma}dx \ \text{ and }\ 
    \T(\sqrt{-\Delta+m^2}\gamma_{\tau_k})=\tau_k\T(\sqrt{-\Delta}\gamma)+O(\tau_k^{-1}).
\end{equation*}
In addition of  $x_j\in \mathcal{Z}$ and $y\in\bar\Lambda$, we have 
\begin{equation*}
\int_{\mathbb{R}^3}V(x)\rho_{\gamma_{\tau_k}}dx=\int_{\mathbb{R}^3}V\big(\frac{x+y}{\tau_k}+x_j\big)\rho_{\gamma}(x)dx=\tau_k^{-p}\int_{\mathbb{R}^3}\frac{V\big(\frac{x+y}{\tau_k}+x_j\big)}{|\frac{x+y}{\tau_k}|^p}|x+y|^p\rho_{\gamma}(x)dx=\tau_{k}^{-p}\big(\iota\bar\kappa+o_k(1)\big).
\end{equation*}
Here the constants $\iota,\bar\kappa>0$ are defined in \eqref{eq-1.17} and \eqref{eq-1.18}, respectively.
Note that $0<p<1$, we take $$\tau_{k}=\Big[(p\iota\bar \kappa)^{-1}\int_{\mathbb{R}^3}(\rho_{\gamma}*|x|^{-1})\rho_{ \gamma}dx\big(\mathcal{K}_\infty^{(N)}-K_k\big)\Big]^{-\frac{1}{p+1}}\to\infty \ \text{as}\  K_k\nearrow \mathcal{K}_\infty^{(N)}$$ to obtain that 
 \begin{equation}\label{1.5.1}
\begin{aligned}
E_{K_k}(N)\le \mathcal{E}_{K_k}(\gamma_{\tau_k})\le& \tau_k\Big(\mathcal{K}_\infty^{(N)}-K_k\Big)\int_{\mathbb{R}^3}(\rho_{\gamma}*|x|^{-1})\rho_{ \gamma}dx+\iota\bar\kappa\tau_{k}^{-p}+o(\tau_k^{-p})\\
=&\big(1+o_k(1)\big)\frac{p+1}{p} (p\iota\bar\kappa)^{\frac{1}{p+1}}\Big(\int_{\mathbb{R}^3}(\rho_{\gamma}*|x|^{-1})\rho_{ \gamma}dx\Big)^{\frac{p}{p+1}}(\mathcal{K}_\infty^{(N)}-K_k)^{\frac{p}{p+1}},
\end{aligned}
\end{equation}
where in the first inequality,  we  used $\T(\sqrt{-\Delta}\gamma)=\int_{\mathbb{R}^3}(\rho_{\gamma}*|x|^{-1})\rho_{ \gamma}dx$ since $\gamma\in \mathcal{S}^1$ is a minimizer of \eqref{1.1}.

Recall $z_k \overset{k}\to z_0\in \Lambda$ and $\varepsilon_k\to 0^+$ are defined in Theorem \ref{th3}. From \eqref{1.5} we see that $z_0=x_s$ for some $1\le s\le l$. We claim that 
\begin{equation}\label{eq-claim}
    p_s=p, \ \ \lim_{k\to\infty}\frac{|z_k-x_s|}{\varepsilon_k}<\infty\ \text{ and }\ \lim_{k\to\infty}\int_{\mathbb{R}^3}\rho_{\tilde\gamma_k}dx=\int_{\mathbb{R}^3}\rho_{\gamma}dx.
\end{equation}
Indeed, from \eqref{1.1}, \eqref{eq-gamma0} and \eqref{eq-zero} we see that 
\begin{equation}\label{eq-kin}
\begin{aligned}
    &\T(\sqrt{-\Delta+m^2\varepsilon_k^2}\tilde \gamma_k)-K_k\int_{\mathbb{R}^3}(\rho_{\tilde \gamma_k}*|x|^{-1})\rho_{\tilde \gamma_k}dx\geq (\mathcal{K}_\infty^{(N)}-K_k)\int_{\mathbb{R}^3}(\rho_{\tilde \gamma_k}*|x|^{-1})\rho_{\tilde \gamma_k}dx\\
    &=\big(1+o_k(1)\big)(\mathcal{K}_\infty^{(N)}-K_k)\int_{\mathbb{R}^3}(\rho_{\gamma}*|x|^{-1})\rho_{ \gamma}dx
\end{aligned}
\end{equation}
and 
\begin{equation}\label{eq-Vss}
\int_{\mathbb{R}^3}V(\varepsilon_kx+z_k)\rho_{\tilde\gamma_{k}}dx=\varepsilon_k^{p_s}\int_{\mathbb{R}^3}\frac{V(\varepsilon_kx+z_k)}{|\varepsilon_kx+z_k-x_s|^{p_s}}\big|x+\frac{z_k-x_s}{\varepsilon_k}\big|^{p_s}\rho_{\tilde \gamma_k}(x)dx.
\end{equation}
Now, assume that \eqref{eq-claim} does not hold, which means
\begin{equation*}
\text{either }p_s<p, \text{ or } \lim_{k\to\infty}\frac{|z_k-x_s|}{\varepsilon_k}=\infty,  \text{ or }  \liminf_{k\to\infty}\int_{\mathbb{R}^3}\rho_{\tilde\gamma_k-\gamma}dx>0.
\end{equation*}
In any one case, we can always deduce from \eqref{eq-gamma0} and \eqref{eq-Vss} that,   for any given $M\gg1$, there holds
\begin{equation*}
\lim_{k\to\infty}\varepsilon_k^{-p}\int_{\mathbb{R}^3}V(\varepsilon_kx+z_k)\rho_{\tilde\gamma_{k}}dx=\lim_{k\to\infty}\varepsilon_k^{p_s-p}\int_{\mathbb{R}^3}\frac{V(\varepsilon_kx+z_k)}{|\varepsilon_kx+z_k-x_s|^{p_s}}\big|x+\frac{z_k-x_s}{\varepsilon_k}\big|^{p_s}\rho_{\tilde \gamma_k}(x)dx\gg M. 
\end{equation*}
This together with \eqref{eq-kin} indicates that 
\begin{equation*}
\begin{aligned}
    E_{K_k}(N)&=\mathcal{E}_{K_k}(\gamma_k)=\varepsilon_k^{-1}\big[\T(\sqrt{-\Delta+m^2\varepsilon_k^2}\tilde \gamma_k)-K_k\int_{\mathbb{R}^3}(\rho_{\tilde \gamma_k}*|x|^{-1})\rho_{\tilde \gamma_k}dx\big]+\int_{\mathbb{R}^3}V(\varepsilon_kx+z_k)\rho_{\tilde\gamma_{k}}dx\\
    &\ge \varepsilon_k^{-1}\big(1+o_k(1)\big)(\mathcal{K}_\infty^{(N)}-K_k)\int_{\mathbb{R}^3}(\rho_{\gamma}*|x|^{-1})\rho_{ \gamma}dx+M\varepsilon_k^{p}\ge C(M) (\mathcal{K}_\infty^{(N)}-K_k)^{\frac{p}{p+1}},
\end{aligned}
\end{equation*}
where the constant $C(M)>0$ satisfies $C(M)\to\infty$ as $M\to\infty$. This however contradicts the previous upper bound  of $E_{K_k}(N)$ in \eqref{1.5.1}. Thus claim \eqref{eq-claim} is proved.  

Because of \eqref{eq-claim}, we may  assume that, up to a subsequence,
\begin{equation*}
    \lim_{k\to\infty}\frac{z_k-x_s}{\varepsilon_k}=y\in\mathbb{R}^3.
\end{equation*}
Applying \eqref{eq-Vss} again, we get that 
\begin{equation}\label{eq-Vss2}
\lim_{k\to\infty}\varepsilon_k^{-p}\int_{\mathbb{R}^3}V(\varepsilon_kx+z_k)\rho_{\tilde\gamma_{k}}dx=\lim_{k\to\infty}\int_{\mathbb{R}^3}\frac{V(\varepsilon_kx+z_k)}{|\varepsilon_kx+z_k-x_s|^{p}}\big|x+\frac{z_k-x_s}{\varepsilon_k}\big|^{p}\rho_{\tilde \gamma_k}(x)dx=\iota_s\int_{\mathbb{R}^3}\big|x+y\big|^{p}\rho_{ \gamma}dx\geq \iota\bar\kappa,
\end{equation}
where the ``=" in the last inequality holds if and only if $x_s\in \mathcal{Z}$ and $y\in\bar\Gamma$. \eqref{eq-kin} and \eqref{eq-Vss2} indicate that
\begin{equation}\label{eq-lower}
\begin{aligned}
    E_{K_k}(N)&=\varepsilon_k^{-1}\big[\T(\sqrt{-\Delta+m^2\varepsilon_k^2}\tilde \gamma_k)-K_k\int_{\mathbb{R}^3}(\rho_{\tilde \gamma_k}*|x|^{-1})\rho_{\tilde \gamma_k}dx\big]+\int_{\mathbb{R}^3}V(\varepsilon_kx+z_k)\rho_{\tilde\gamma_{k}}dx\\
    &\ge \varepsilon_k^{-1}\big(1+o_k(1)\big)(\mathcal{K}_\infty^{(N)}-K_k)\int_{\mathbb{R}^3}(\rho_{\gamma}*|x|^{-1})\rho_{ \gamma}dx+\iota\bar\kappa\varepsilon_k^{p} \\
    &\ge \big(1+o_k(1)\big)\frac{p+1}{p} (p\iota\bar\kappa)^{\frac{1}{p+1}}\Big(\int_{\mathbb{R}^3}(\rho_{\gamma}*|x|^{-1})\rho_{ \gamma}dx\Big)^{\frac{p}{p+1}}(\mathcal{K}_\infty^{(N)}-K_k)^{\frac{p}{p+1}},
\end{aligned}
\end{equation}
where the ``="  in the last inequality holds if and only if 
$$\varepsilon_{k}=\big(1+o_k(1)\big)\Big[(p\iota\bar \kappa)^{-1}\int_{\mathbb{R}^3}(\rho_{\gamma}*|x|^{-1})\rho_{ \gamma}dx\big(\mathcal{K}_\infty^{(N)}-K_k\big)\Big]^{\frac{1}{p+1}}\to0^+ \ \text{as}\  K_k\nearrow \mathcal{K}_\infty^{(N)}.$$

Comparing the lower bound in \eqref{eq-lower} with the upper bound in 
\eqref{1.5.1}, we see that \eqref{eq-energy} holds. This further indicates that all equalities in \eqref{eq-Vss2} and \eqref{eq-lower} hold, and thus \eqref{eq-vas} and \eqref{eq-points} is proved. Finally, from \eqref{eq-claim} we know that 
$$r\equiv\int_{\mathbb{R}^3}\rho_{\tilde\gamma_k}dx=\int_{\mathbb{R}^3}\rho_{\gamma}dx=R$$ and  then \eqref{eq-zero2}  follows.  The proof of Theorem \ref{th4} is finished.
 \qed

\noindent{\bf Acknowledgments.} We would like to express our gratitude to Prof. David Gontier for his valuable clarification on the technical details of their research. This work was supported by the National Natural Science Foundation of China under grant number 12322106, 12171379, 12271417 and supported by   ``the Fundamental Research Funds for the Central Universities”.

\section{Appendix}
In this section, we will present the polynomial decay of the solution of system as follows
\begin{equation}\label{A.1}
\sqrt{-\Delta }u_i-\frac{2}{\alpha}(\rho*|x|^{-\alpha})u_i=\mu_i u_i,\quad 1\le i\le N,
\end{equation}
where $0<\alpha\le 1$, $\rho=\sum_{i=1}^N \beta_i|u_i|^2$ with $\beta_i>0$ and $\mu_i<0$, for all $1\le i\le N$.
\begin{theorem}\label{tha}
Assume that $(u_1,...,u_N)$ is the solution of system (\ref{A.1}). Then we have
\begin{equation*}
|u_i(x)|\le \frac{C}{1+|x|^4},\quad \forall\ 1\le i\le N, \ \text{ and }\ \int_{\mathbb{R}^3}\frac{\rho(y)}{|x-y|^{\alpha}}dy\le \frac{C}{1+|x|^{\alpha}}.
\end{equation*}
Furthermore, $u_1$ corresponding to the first negative eigenvalue $\mu_1$ satisfies
\begin{equation*}
u_1(x)\ge \frac{C}{1+|x|^4}.
\end{equation*}
\end{theorem}

This proof follows the ideas from \cite{fl,le}. Before approaching theorem \ref{tha}, we first pose the regularity of solutions.

\begin{lemma}\label{A.2}
Assume $F=\frac{2}{\alpha}\rho*|x|^{-\alpha}$ be a mapping, where $\rho=\sum_{i=1}^N \beta_i|u_i|^2$, $F$ maps $H^\frac{1}{2}(\mathbb{R}^3)$ into itself.
\end{lemma}
\begin{proof}
By H\"older inequality, Hardy-Kato inequality and Theorem 2.5 in \cite{h},
\begin{equation*}
\|(\rho*|x|^{-\alpha})u\|_2\le \|\rho*|x|^{-\alpha}\|_{\infty}\|u\|_2\le C\|(-\Delta)^{\frac{\alpha}{4}}\sqrt \rho\|_2\|u\|_2\le C\|(-\Delta)^{\frac{1}{4}}\sqrt \rho\|_2\|u\|_2\le C\|u\|_2.
\end{equation*}
In addition, since $0<\frac{3}{3-\alpha}<\frac{3}{2}$,, we deduce from  the generalized Leibniz rule that
\begin{equation*}
\begin{aligned}
\|(-\Delta)^{\frac{1}{4}}[(\rho*|x|^{-\alpha})u]\|_2&\le C\|(-\Delta)^{\frac{1}{4}}(\rho*|x|^{-\alpha})\|_6 \|u\|_3+\|\rho* |x|^{-\alpha}\|_{\infty} \|(-\Delta)^{\frac{1}{4}}u\|_2\\
&\le C\|(-\Delta)^{\frac{1}{4}-\frac{3-\alpha}{2}}\rho\|_6\|u\|_3+C\|(-\Delta)^{\frac{1}{4}}\sqrt{\rho}\|_2\|u\|_{H^\frac{1}{2}}\\
&\le C\| |x|^{-\frac{1+2\alpha}{2}}*\rho\|_6\|u\|_{H^\frac{1}{2}}+C\|u\|_{H^\frac{1}{2}}\\
&\le C\|\rho\|_{\frac{3}{3-\alpha}}\|u\|_{H^\frac{1}{2}}+C\|u\|_{H^\frac{1}{2}}\le C\|u\|_{H^\frac{1}{2}}.
\end{aligned}
\end{equation*}
\end{proof}

Assume that $u_i$ satisfies \eqref{A.1} with $\mu_i<0$, i.e.,
\begin{equation*}
u_i=(\sqrt{-\Delta}-\mu_i)^{-1}\Big[\frac{2}{\alpha}(\rho*|x|^{-\alpha})u_i\Big]\in H^{\frac{3}{2}}(\mathbb{R}^3),
\end{equation*}
 it then follows from inequality of Bessel's operator and Lemma \ref{A.2},
\begin{equation*}
\rho=\sum_{i=1}^N\beta_i|u_i|^2\in W^{\frac{3}{2},1}(\mathbb{R}^3).
\end{equation*}

\begin{lemma}
$F$ defined in Lemma \ref{A.2} maps $H^\frac{3}{2}(\mathbb{R}^3)$ into itself.
\end{lemma}
\begin{proof} We just  need to check that $(-\Delta)^\frac{3}{4}(\rho*|x|^{-\alpha})u\in L^2(\mathbb{R}^3)$, indeed,
\begin{equation*}
\begin{aligned}
\|(-\Delta)^\frac{3}{4}(\rho*|x|^{-\alpha})u\|_2&\le C\|(-\Delta)^{\frac{3}{4}}(\rho*|x|^{-\alpha})\|_6 \|u\|_3+\|\rho* |x|^{-\alpha}\|_{\infty} \|(-\Delta)^{\frac{3}{4}}u\|_2\\
&\le C\|(-\Delta)^{\frac{3}{4}-\frac{3-\alpha}{2}}\rho\|_6\|u\|_3+C\|(-\Delta)^{\frac{1}{4}}\sqrt{\rho}\|_2\|u\|_{H^\frac{3}{2}}\\
&\le C\|(-\Delta)^{\frac{2\alpha-1}{4}}\rho\|_2\|u\|_{H^\frac{3}{2}}+C\|u\|_{H^\frac{3}{2}}.
\end{aligned}
\end{equation*}
If $2\alpha-1\ge 0$, then $$\|(-\Delta)^{\frac{2\alpha-1}{4}}\rho\|_2\le C\|(-\Delta)^{\frac{2\alpha-1}{4}}\sqrt{\rho}\|_6\|\sqrt{\rho}\|_3\le C\|(-\Delta)^{\frac{2\alpha+1}{4}}\sqrt{\rho}\|_2\le C\sum_{i=1}^N\beta_i\|u_i\|_{H^{\frac{3}{4}}}^2\le C,$$
otherwise, we obtain from $\frac{3}{2}<\frac{3}{2-\alpha}<2$ that 
\begin{equation*}
\begin{aligned}
\|(-\Delta)^{\frac{2\alpha-1}{4}}\rho\|_2&\le C\|\rho*|x|^{-\frac{5+2\alpha}{2}} \|_2\le C\|\rho\|_{\frac{3}{2-\alpha}}
\le C\|\rho\|_{W^{\frac{3}{2},1}}\le C.
\end{aligned}
\end{equation*}
Therefore, there holds that 
\begin{equation*}
\|(-\Delta)^\frac{3}{4}(\rho*|x|^{-\alpha})u\|_2\le C\|u\|_{H^\frac{3}{2}}.
\end{equation*}
\end{proof}

\noindent\textit{Proof of Theorem \ref{A.1}.}
From the argument above, we can iterate that the solution of (\ref{A.1}) $u_i\in H^{\frac{5}{2}}(\mathbb{R}^3)$, and then $u_i\in C^{1,\beta}$ for some $\beta>0$ by applying  Sobolev's inequalities. This implies that  $|u_i(x)|\to 0$ as $|x|\to \infty$ and $(\rho*|x|^{-\alpha})(x)\to 0$ as $|x|\to \infty$. From proposition IV.1 in \cite{cms}, we deduce
\begin{equation*}
|u_i(x)|\le \frac{C}{1+|x|^4}.
\end{equation*}
By Newton's theorem and $\rho=\sum_{i=1}^N\beta_iu_i^2$, we get
\begin{equation*}
\rho*|x|^{-\alpha}\le \frac{1}{1+|x|^{\alpha}}\int_{\mathbb{R}^3}\frac{1}{1+|x|^8} dx\le \frac{C}{1+|x|^{\alpha}}.
\end{equation*}
In addition, applying Proposition IV.3 in \cite{cms} to $u_1$, we obtain
\begin{equation*}
u_1(x)\ge \frac{C}{1+|x|^4}.
\end{equation*}\qed

\begin{theorem}
Let $(u_1,...,u_N)$ be a solution of system (\ref{A.1}), then for $1\le i\le N$, $u_i$ satisfies
\begin{equation}\label{eq-Poho1}
\int_{\mathbb{R}^3}|(-\Delta)^{\frac{1}{4}}u_i|^2dx=\frac{6-\alpha}{\alpha}\int_{\mathbb{R}^3}(|x|^{-\alpha}*\rho)u_i^2dx+\frac{3}{2}\int_{\mathbb{R}^3}\mu_iu_i^2dx+\frac{1}{\alpha}\int_{\mathbb{R}^3}|x|^{-\alpha}*(x\cdot \nabla \rho)u_i^2dx.
\end{equation}
Moreover, system (\ref{A.1}) satisfies
\begin{equation}\label{eq-Poho2}
{\rm Tr}(\sqrt{-\Delta}\gamma)=\frac{6-\alpha}{2\alpha}\int_{\mathbb{R}^3}(|x|^{-\alpha}*\rho)\rho dx+\frac{3}{2}\sum_{i=1}^N\int_{\mathbb{R}^3}\mu_i \alpha_i u_i^2dx
\end{equation}
and
\begin{equation*}
{\rm Tr}(\sqrt{-\Delta}\gamma)=\int_{\mathbb{R}^3}(\rho*|x|^{-\alpha})\rho dx.
\end{equation*}
\end{theorem}
\begin{proof}
According to Theorem \ref{tha}, we know that $u_i\in C^{1,\beta_1}(\mathbb{R}^3)$ $(0<\beta_1<1)$. By similar arguments as in Lemma 6 in \cite{psv}, we obtain $u_i\in C^{2,\beta_2}(\mathbb{R}^3)$ ( $0<\beta_2<1$). From the classical results of 
\cite{cs}, we can  transform the nonlocal problem (\ref{A.1}) into local problem as follows
\begin{equation}\label{B1}
\begin{cases}
-\Delta w(x,y))=0,\quad &\text{in}\  \mathbb{R}^3\times\{y\ge0\},\\
\frac{\partial w(x,y)}{\partial \nu}=(\frac{2}{\alpha}|x|^{-\alpha}*\rho)w+\mu_i w,\quad &\text{on}\  \mathbb{R}^3\times\{y=0\},\\
w(x,0)=u_i(x),\quad &\text{on}\  \mathbb{R}^3\times\{y=0\}.
\end{cases}
\end{equation}

From basic theory of Harmonic function, we know $w$ is uniquely determined by boundary value $u$ and $w\in C^2(\mathbb{R}^4_+)$ due to $u\in C^{2,\beta_2}(\mathbb{R}^3)$. Then, we follow the ideas from Proposition 4.1 in \cite{cw} to define
\begin{equation*}
D_{R,\delta}^+=\{z=(x,y)\in \mathbb{R}^3\times [\delta, +\infty): |z|^2\le R^2\},
\end{equation*}
and its boundary
\begin{equation*}
\begin{aligned}
\partial D_{R,\delta}^1&=\{z=(x,y)\in \mathbb{R}^3\times \{y=\delta\}: |x|^2\le R^2-\delta^2\},\\
\partial D_{R,\delta}^2&=\{z=(x,y)\in \mathbb{R}^3\times [\delta,\infty): |z|^2= R^2\}.
\end{aligned}
\end{equation*}
Then we have from (\ref{B1}) that
\begin{equation}\label{B2}
\begin{aligned}
0=&\int_{D_{R,\delta}^+}\Delta w(x,y)dxdy\\
=&\int_{\partial D_{R,\delta}^1}\Big[(x,\nabla_x w)\frac{\partial w(x,y)}{\partial \nu}+\frac{y}{2}|\nabla w|^2\Big]d\sigma+\int_{\partial D_{R,\delta}^2}\Big[\frac{1}{R}|(z,\nabla w)|^2
-\frac{R}{2}|\nabla w|^2\Big]d\sigma +\int_{D_{R,\delta}^+}|\nabla w|^2dz\\
=&I+II+III.
\end{aligned}
\end{equation}
Letting $\delta \to 0$, we get
\begin{equation*}
\begin{aligned}
\lim_{\delta\to 0}&\int_{\partial D_{R,\delta}^1}\Big[(x,\nabla_x w)\frac{\partial w(x,y)}{\partial \nu}\Big]d\sigma 
= \int_{B_R}(\frac{2}{\alpha}|x|^{-\alpha}*\rho u_i+\mu_i u_i)(x,\nabla u_i)dx\\
=&\int_{B_R}\Big[\frac{1}{2}{\rm div}(\frac{2}{\alpha}|x|^{-\alpha}*\rho xu_i^2+\mu_i xu_i^2)
-\frac{1}{\alpha}x\cdot \nabla (|x|^{-\alpha}*\rho)u_i^2-\frac{3}{\alpha}(|x|^{-\alpha}*\rho)u_i^2-\frac{3}{2}\mu_i u_i^2\Big]dx\\
=&R\int_{\partial B_R}\big[\frac{1}{\alpha}(|x|^{-\alpha}*\rho)u_i^2+\frac{1}{2}\mu_i u_i^2\big]dS-\frac{1}{2}\int_{B_R}\Big[3\mu_i u_i^2
+\frac{6}{\alpha}(|x|^{-\alpha}*\rho)u_i^2+\frac{2}{\alpha}\Big(x\cdot \nabla (|x|^{-\alpha}*\rho)\Big) u_i^2\Big]dx.
\end{aligned}
\end{equation*}
It follows from
$\lim_{\delta\to 0}\int_{\partial D_{R,\delta}^1}\frac{y}{2}|\nabla w|^2d\sigma=0$
that
\begin{equation*}
\begin{aligned}
I=&R\int_{\partial B_R}\big[\frac{1}{\alpha}(|x|^{-\alpha}*\rho)u_i^2+\frac{1}{2}\mu_i u_i^2\big]dS
-\frac{1}{2}\int_{B_R}\Big[3\mu_i u_i^2
+\frac{6}{\alpha}(|x|^{-\alpha}*\rho)u_i^2+\frac{2}{\alpha}(x\cdot \nabla (|x|^{-\alpha}*\rho)) u_i^2\Big]dx.
\end{aligned}
\end{equation*}

Next, we claim that there exists a consequence $\{R_n\}$ satisfying $R_n\to \infty$ as $n\to \infty$, such that
\begin{equation*}
\lim_{n\to \infty}R_n\int_{\partial B_{R_n}}\Big[\frac{2}{\alpha}(|x|^{-\alpha}*\rho)u_i^2+\mu_iu_i^2\Big]dS=0,
\end{equation*}
and
\begin{equation}\label{B5}
\lim_{n\to \infty} \int_{\partial D_{R_n,\delta}^2}[\frac{1}{R_n}|(z,\nabla w)|^2-\frac{R_n}{2}|\nabla w|^2]d\sigma =0,\quad \forall \delta>0.
\end{equation}
 Set $G_i:=\frac{2}{\alpha}(|x|^{-\alpha}*\rho)u_i^2+\mu_iu_i^2
$, it suffices to prove that
\begin{equation}\label{eq-B6}
\lim_{n\to \infty}R_n\int_{\partial B_{R_n}}|G_i|dS=\lim_{n\to \infty}R_n \int_{\partial D_{R_n,\delta}^2}|\nabla w|^2d\sigma=0,
\end{equation}
since
\begin{equation*}
\int_{\partial B_{R_n}}|G_i|dS\ge \int_{\partial B_{R_n}}G_idS,
\end{equation*}
and
\begin{equation*}
\begin{aligned}
\int_{\partial D_{R_n,\delta}^2}[\frac{1}{R_n}|(z,\nabla w)|^2-\frac{R_n}{2}|\nabla w|^2]d\sigma \le& \int_{\partial D_{R_n,\delta}^2} (\frac{|z|^2|\nabla w|^2}{R_n})-\frac{R_n|\nabla w|^2}{2}d\sigma
= \frac{R_n}{2}\int_{\partial D_{R_n,\delta}^2}|\nabla w|^2d\sigma.
\end{aligned}
\end{equation*}
If \eqref{eq-B6} fails, we  may assume that 
\begin{equation*}
\liminf_{R\to \infty} R\int_{\partial B_R}|G_i|d\sigma =\tau>0.
\end{equation*}
There exists $R_1>0$ large enough, such that for all $R\ge R_1$,
\begin{equation*}
\int_{\partial B_R}|G_i|d\sigma \ge \frac{\tau}{R},
\end{equation*}
which implies that
\begin{equation*}
\int_{\mathbb{R}^3}|G_i|dx\ge \int_{R_1}^{\infty}\frac{\tau}{R}dR\to \infty,\quad as\ R\to\infty.
\end{equation*}
This contradicts the facts that  $\sqrt{\rho}, u_i\in H^{\frac{1}{2}}(\mathbb{R}^3)$. Similarly,  we can prove that $\lim\limits_{n\to \infty}R_n \int_{\partial D_{R_n,\delta}^2}|\nabla w|^2d\sigma=0$ because $w\in \dot{H}^{1}(\mathbb{R}_+^{N+1})$. Claim \eqref{B5} is proved.

By (\ref{B2})-(\ref{B5}), we deduce that
\begin{equation*}
\int_{\mathbb{R}^3}|(-\Delta)^{\frac{1}{4}}u_i|^2dx=\frac{3}{2}\int_{\mathbb{R}^3}\Big[\frac{2}{\alpha}(|x|^{-\alpha}*\rho)u_i^2+\mu_i u_i^2\Big]dx+ \frac{1}{\alpha}\int_{\mathbb{R}^3}(x\cdot \nabla(|x|^{-\alpha}*\rho)u_i^2dx.
\end{equation*}
On the other hand, it follows from \cite{gz1} that 
\begin{equation}\label{B6}
\begin{aligned}
x\cdot \nabla (|x|^{-\alpha}*\rho)=x\cdot |x|^{-\alpha}*\nabla \rho
=|x|^{-\alpha}*(x\cdot \nabla\rho)+(3-\alpha)|x|^{-\alpha}*\rho,
\end{aligned}
\end{equation}
then we obtain \eqref{eq-Poho1}.
Moreover, by $\rho=\sum_{i=1}^N\beta_i|u_i|^2$, we know
\begin{equation*}
\text{Tr}(\sqrt{-\Delta}\gamma)=\frac{6-\alpha}{\alpha}\int_{\mathbb{R}^3}(|x|^{-\alpha}*\rho)\rho dx+\frac{3}{2}\sum_{i=1}^N\int_{\mathbb{R}^3}\mu_i\beta_i u_i^2dx +\frac{1}{\alpha}\int_{\mathbb{R}^3}|x|^{-\alpha}*(x\cdot \nabla \rho)\rho dx.
\end{equation*}
Since
\begin{equation*}
\begin{aligned}
\int_{\mathbb{R}^3}|x|^{-\alpha}*(x\cdot \nabla \rho)\rho dx
=&\int_{\mathbb{R}^3}(|x|^{-\alpha}*\rho)x\cdot \nabla \rho dx
=\int_{\mathbb{R}^3}(|x|^{-\alpha}*\rho)\rho dx-\int_{\mathbb{R}^3}\Big((|x|^{-\alpha}*\nabla \rho)\cdot x\Big) \rho dx\\
=&-\frac{6-\alpha}{2}\int_{\mathbb{R}^3}(|x|^{-\alpha}*\rho)\rho dx,
\end{aligned}
\end{equation*}
where we used (\ref{B6}) in  the last equality. The above two estimates indicate \eqref{eq-Poho2}. Furthermore, (\ref{A.1}) implies that
\begin{equation*}
\text{Tr}(\sqrt{-\Delta}\gamma)=\frac{2}{\alpha}\int_{\mathbb{R}^3}(|x|^{-\alpha}*\rho)\rho dx+\sum_{i=1}^N\int_{\mathbb{R}^3}\mu_i\beta_i u_i^2dx
\ \text{ and }\ 
\text{Tr}(\sqrt{-\Delta}\gamma)=\int_{\mathbb{R}^3}(\rho*|x|^{-\alpha})\rho dx.
\end{equation*}
\end{proof}

\end{document}